\documentclass[9pt,a4paper,twocolumn]{IEEEtran}

\usepackage{amsthm}

\usepackage{soul,xcolor}
\setstcolor{red}
\usepackage{cite}
\usepackage{amsmath,amssymb,amsfonts}
\usepackage{algorithmic}
\usepackage{textcomp}
\usepackage{geometry}
\usepackage{graphicx}
\usepackage{color}
\DeclareMathDelimiter{\orbrack}{\mathopen}{operators}{"5D}{largesymbols}{"03}
\DeclareMathDelimiter{\clbrack}{\mathclose}{operators}{"5B}{largesymbols}{"02}
\def\intcc#1{\ensuremath{[#1]}}
\def\intoo#1{\ensuremath{\orbrack#1\clbrack}}
\def\intoc#1{\ensuremath{\orbrack#1]}}
\def\intco#1{\ensuremath{[#1\clbrack}}

\def\defas{\ensuremath{\mathrel{:=}}}

\DeclareMathOperator{\id}{id}

\def\Set#1#2{\ensuremath{
\left\{#1\,\middle|\,#2\right\}
}}

\def\e{\mathrm{e}}

\def\Zonotope#1#2{\ensuremath{\mathcal{Z}(#1,#2)}}

\def\Gen#1{\ensuremath{\mathrm{Gen}(#1)}}

\def\norm#1{\|  #1 \| }

\def\abs#1{\left|  #1 \right| }

\newtheorem{theorem}{Theorem}
\newtheorem{lemma}[theorem]{Lemma}

\newtheorem{remark}{Remark}

\newtheorem{corollary}{Corollary}
\newtheorem{definition}{Definition}

\title{ Under-Approximate   Reachability Analysis for a Class of Linear Systems with Inputs}% with Convergence Guarantees} %Using Truncated Taylor Expansions of the Matrix Exponential} 
\author{Mohamed Serry and  Jun Liu \thanks{Mohamed Serry is with the Department of Mechanical and Mechatronics Engineering, University of Waterloo, Waterloo, Ontario, Canada (email: mserry@uwaterloo.ca).}  
\thanks{Jun Liu is with the Department of Applied Mathematics, University of Waterloo, Waterloo, Ontario, Canada  (e-mail: j.liu@uwaterloo.ca).} \thanks{This work was funded by NSERC DG, CRC, and ERA programs.}  }   
\date{}   
\begin{document}
\maketitle
\begin{abstract}              
Under-approximations of reachable sets and tubes have been receiving  growing research attention due to their important roles in control synthesis and verification.  Available under-approximation methods applicable to continuous-time linear  systems  typically assume the ability to compute transition matrices and their integrals exactly, which is not feasible in general, and/or suffer from high computational costs. In this note, we attempt to overcome these drawbacks for a class of linear time-invariant (LTI) systems, where we propose a novel method to under-approximate finite-time forward reachable sets and tubes, utilizing approximations of the matrix exponential and its integral. In particular, we  consider the class  of continuous-time LTI systems with an identity input matrix and  initial and input values belonging to full dimensional sets that are affine transformations of closed unit balls. The proposed method yields computationally efficient under-approximations of reachable sets and tubes, when implemented using zonotopes,  with first-order convergence  guarantees  in the sense of the Hausdorff distance. To illustrate its  performance,  we implement our approach in three numerical examples, where linear systems of dimensions ranging between 2 and 200 are considered.  
\end{abstract}
\section{Introduction}
Reachable sets and tubes of dynamical systems 
 are  central in control synthesis and verification applications, especially in the presence of uncertainties and { constraints} \cite{i14sym,Althoff10}.  Mere approximations of reachable sets and tubes are not sufficient in such frameworks. Instead,  conservative estimations, i.e., over (outer)-approximations,  are typically utilized to ensure all possible behaviors of a given control system are accounted for, which  explains the sheer number of over-approximation methods in the literature  \cite{AsarinDangFrehseGirardLeGuernicMaler06,AlthoffFrehseGirard21}.    

 { In the last few years, there has been a growing interest in additionally   under-approximating reachable sets and tubes for synthesis and verification (see, e.g., \cite{ SheLi20,YinArcakPackardSeiler21,YangOzay21,GoubaultPutot20}), because under-approximations can be used in   estimating  subsets  of the  states that are  attainable under given control constraints \cite{GirardLeGuernicMaler06}, obtaining  subsets  of the initial states  from which all trajectories fulfill safety and reachability specifications   \cite{XueFranzleZhan19}, solving  falsification problems: verifying if reachable sets/tubes intersect with unsafe sets \cite{BhatiaFrazzoli04}, and  measuring the accuracy of computed over-approximations. Motivated by the aforementioned applications, we aim in this paper at investigating under-approximations of forward reachable sets and tubes for continuous-time LTI systems with uncertainties or constraints on the initial and input values, where we attempt to  overcome some of the limitations associated with available methods in the literature. }

%\st{The ellipsoidal method in} \cite{HamadehGoncalves08} \st{was proposed for piecewise affine linear systems, which is derived from linear matrix inequality (LMI) formulations, where under-approximations of projections of reachable sets on switching surfaces are estimated. The method does not account for uncertainties in the input, assumes the ability to obtain solutions to linear systems exactly, and lacks convergence guarantees.} 

Optimal control-based  polytopic approaches \cite{PecsvaradiNarendra71,varaiya2000reach} were proposed for linear systems with uncertain inputs and initial conditions. These methods rely on obtaining boundary points of reachable sets, associated with specified direction vectors,  and then computing the convex hull of the obtained boundary points.  Given a reachable set, the convergence of the polytopic approaches requires computing an increasing number of boundary points until the whole reachable set boundary is obtained, which is computationally expensive, especially if the dimension of the  reachable set is high. An ellipsoidal method was proposed in \cite{KurzhanskiVaraiya00b} for controllable linear systems with ellipsoidal initial and input sets. The method relies on solving initial value problems, derived from maximal principles similar to those presented in \cite{varaiya2000reach}, to obtain ellipsoidal subsets that touch a given reachable set at some boundary points, depending on specified direction vectors. The accuracy of the ellipsoidal method of \cite{KurzhanskiVaraiya00b} in under-approximating a given reachable set is increased by evaluating an increasing number of ellipsoids, which necessitates considering an increasing number of direction vectors, and then taking their union. The aforementioned optimal control-based approaches \cite{PecsvaradiNarendra71,varaiya2000reach,KurzhanskiVaraiya00b} assume the ability to compute transition matrices  and their  integrals  exactly, and that is not feasible in general. In addition, when under-approximating a reachable tube, the mentioned approaches use non-convex representations for the under-approximations, which are challenging to analyze in the contexts of verification and synthesis (see the introduction in \cite{SerryReissig21}).

In \cite{Serry21},  a set-propagation technique was proposed,  yielding convergent  under-approximations of forward reachable sets and tubes for a  general class of linear systems with uncertainties or constraints on the initial and input values, where the under-approximations of reachable sets and tubes  are given as  convex sets and  finite unions of convex  sets, respectively,  which can be analysed with relative ease. However, the method in \cite{Serry21}, like the approaches in \cite{HamadehGoncalves08,PecsvaradiNarendra71,varaiya2000reach,KurzhanskiVaraiya00b}, assumes the ability to compute transition matrices and their integrals exactly. A similar set-propagation approach was proposed in \cite{FaureCieslakHenryVerhaegenAnkersen21} for LTI systems, which relies on  computing Minkowski differences to under-approximate reachable tubes. The method in \cite{FaureCieslakHenryVerhaegenAnkersen21} also suffers from the issue of the method in \cite{Serry21}, { in addition to the computational hurdle of evaluating Minkowski differences (see, e.g., \cite{YangZhangJeanninOzay22})}.

{
The issue of evaluating transition matrices and their integrals exactly can be solved by adopting formally correct under-approximation methods that are designed for general nonlinear systems. For example,  a formally correct interval arithmetic under-approximation approach was proposed in \cite{GoubaultPutot20}; however, such a method  lacks convergence guarantees, and may produce empty under-approximations. 
In \cite{KochdumperAlthoff20},  a novel method was proposed for nonlinear systems with uncertain initial conditions that depends on computing over-approximations and then scaling them down to obtain under-approximations. A drawback of the approach in \cite{KochdumperAlthoff20} is that the scaling necessitates solving  (sub-optimal) optimization problems that involve enclosures  of boundaries of reachable sets, which can be computationally expensive  (see the comparison in Section \ref{sec:Comparison}). Finally, a  recent approach has been proposed in \cite{ShafaOrnik22} to under-approximate reachable sets  when the system parameters (e.g. system matrices) are not known exactly , where collected trajectories (i.e., data) are utilized to estimate the system dynamics. Such an approach is highly valuable in applications when  system identification cannot be attained, due, e.g., to failure or damage mid-operation. However, this approach is conservative (i.e., convergence cannot be attained in general)  as it considers the set of all systems that can  generate the collected trajectories, while fulfilling some specified  assumptions.}

In this work, we present a novel efficient approach that results in   under-approximations of forward finite-time reachable sets and tubes for a class of  LTI  systems with inputs, where approximations of the matrix exponential  and its integral are used, first-order convergence guarantees are provided,  and approximations of  reachable sets and tubes are given as convex sets and finite unions of convex sets, respectively. Our approach is fundamentally based on set-based recursive relations (see, e.g., \cite{Veliov92,Serry21}), where truncation errors are accounted for in an under-approximating manner, utilizing matrix lower bounds \cite{Grcar10}.

%The organization of this paper is as follows. The mathematical preliminaries and notations needed for this work are discussed in Section \ref{sec:Preliminaries}. The class of systems under consideration and problem formulation are introduced in Section \ref{sec:SystemDescription}. Section \ref{sec:ProposedMethod} discusses thoroughly the derivation of the proposed method. The implementation and space complexity of the proposed method are discussed in Section \ref{sec:Implementation}. The convergence guarantees are illustrated in Section \ref{sec:Convergence}. The performance of the proposed method is displayed through three numerical examples in Section \ref{sec:NumericalExample}, and the study is concluded in Section \ref{sec:Conclusion}.

\section{Preliminaries}
\label{sec:Preliminaries}
Let $\mathbb{R}$, $\mathbb{R}_+$, $\mathbb{Z}$, $\mathbb{Z}_{+}$, and $\mathbb{C}$ denote
the sets of real numbers, non-negative real numbers, integers,
non-negative integers, and complex numbers, respectively, and
$\mathbb{N} \defas \mathbb{Z}_{+} \setminus \{ 0 \}$.
Let $\intcc{a,b}$, $\intoo{a,b}$,
$\intco{a,b}$, and $\intoc{a,b}$
denote closed, open, and half-open
intervals, respectively, with end points $a$ and $b$, and
 $\intcc{a;b}$, $\intoo{a;b}$,
$\intco{a;b}$, and $\intoc{a;b}$ stand for their discrete counterparts,
e.g.,~$\intcc{a;b} = \intcc{a,b} \cap \mathbb{Z}$, and
$\intco{1;4} = \{ 1,2,3 \}$.
Given any map $f \colon A \to B$, the image of a subset $C \subseteq
A$ under $f$ is given by $f(C) = \Set{ f(c) }{ c \in C }$.
The  identity map $X \to X \colon x \mapsto x$ is denoted by $\id$,
where the domain of definition $X$ will always be
clear form the context.
%%%%%%%%%%%%%%%%%%%%%%%%%%%%%%%%%%%%%%%%%%%%%%%%%%%%%%%%%%%%%%%%%%%%%%
%Arithmetic operations involving subsets of $\mathbb{R}^{n}$ are
%defined pointwise, e.g.
The Minkowski sum of  $M, N \subseteq \mathbb{R}^{n}$ is defined as
$M + N \defas \Set{ y + z }{ y \in M, z \in N }$.
By $\| \cdot \|$ we denote any norm on $\mathbb{R}^{n}$, the norm of a non-empty
subset $M \subseteq \mathbb{R}^{n}$ is defined by
$\| M \| \defas \sup_{x \in M} \| x \|$,   $\mathbf{B}_{n} \subseteq \mathbb{R}^{n}$ is the closed unit ball w.r.t.~$\| \cdot \|$, and the maximum norm
% ``maximum norm'': AmannEscher98Vol1, Ch II.3
on $\mathbb{R}^n$ is denoted by $\| \cdot \|_{\infty}$ ($\| x \|_{\infty} = \max \Set{ | x_i | }{ i \in \intcc{1;n} },~x \in \mathbb{R}^n$).   Given norms on $\mathbb{R}^n$ and $\mathbb{R}^m$, $\mathbb{R}^{n \times m}$ is endowed with the
usual matrix norm,
$\| A \| = \sup_{\| x \| \le 1} \| A x \|$ for
$A \in \mathbb{R}^{n \times m}$, e.g.,  the matrix norm of $A$ induced by the maximum norm is $\norm{A}_{\infty}=\max_{i\in \intcc{1;n}}\sum_{j=1}^{m}\abs{A_{i,j}}$.  Given a square matrix $A\in \mathbb{R}^{n\times n}$,  $\rho(A)$ denotes the  spectral radius of $A$, i.e.,  $\rho(A)\defas \max \{\abs{\lambda},~\lambda \textrm{ is an eigenvalue of }A \}$. 
The spectral radius satisfies the  property below, which follows from  \cite[proof~of~Lemma~5.6.10,~p.~348]{HornJohnson12}).
\begin{lemma}\label{lem:SpectralRadius}
Let $A\in \mathbb{R}^{n\times n}$. For each $\epsilon>0$, there exists an induced matrix norm $\norm{\cdot}_{\epsilon}$ such that  $\norm{A}_{\epsilon}\leq \rho(A)+\epsilon$. 
%\footnote{If we extend the definition of matrix $A$ to be a mapping on $\mathbb{C}^{n}$, equip $\mathbb{C}^{n}$ with any arbitrary norm $\mathcal{N}(\cdot)$, and define $\mathcal{N}_{o}(\cdot)$ to be the matrix norm induced by $\mathcal{N}$, then $\rho(A)\leq \mathcal{N}_{o}(A)$ \cite[Theorem~5.6.9,~p.~347]{HornJohnson12}. Moreover, the statement of Lemma \ref{lem:SpectralRadius} still holds.  This implies that, when extending the definition of $A$, $\rho(A)$ becomes the ``greatest lower bound'' of all induced matrix norms of $A$, which motivates using the spectral radius in this work.}.
\end{lemma}
Given $A\in \mathbb{R}^{n\times m}$, $\mathrm{rank}(A)$, $\mathrm{col}(A)$, and $A^{\dagger}\in \mathbb{R}^{m\times n}$   denote the rank,  the column space, and the Moore–Penrose inverse of $A$, respectively ($A^{\dagger}=A^{-1}$ if $A$ is invertible). %The matrix inverse operator $A \rightarrow A^{-1}$ is continuous at invertible matrices  as seen in the lemma below:  
The following lemma provides a sufficient condition to check invertibility. %
%We utilizes the following lemma for checking invertibility. 
\begin{lemma}\label{Lem:ContinuityOfMatrixInverse}
Let $P,\tilde{P}\in \mathbb{R}^{n\times n}$, and $P$ be invertible. If  $\norm{P-\tilde{P}}\norm{{P}^{-1}}<1$, then $\tilde{P}$ is invertible.
\end{lemma}
\begin{proof}
See, e.g., \cite[proof of Theorem 5.7,~p.~111]{Maccluer09}.
% The matrix $\tilde{P}$ can be written as
% $
% \tilde{P}=P(\id-P^{-1}(P-\tilde{P})).
% $
% We have $P$ is invertible, by assumption, and $\id-P^{-1}(P-\tilde{P})$ is invertible as $\norm{P^{-1}(P-\tilde{P})}\leq \norm{\tilde{P}^{-1}}\norm{P-\tilde{P}}<1$, which completes the proof.
\end{proof}
%, which, for the case  $\mathrm{rank}(A)=n$, corresponds to the right inverse of $A$, i.e., $A A^{\dagger}=\id$, and in that case, it    is given explicitly as $A^{\dagger}=A^{\intercal}(AA^{\intercal})^{-1}$. 
   Given $A\in \mathbb{R}^{n\times m}\setminus \{0\}$, $\norm{A}_{l}$ denotes the matrix lower bound of $A$ w.r.t.  $\norm{\cdot}$, which is defined as
$
\norm{A}_{l}\defas \max \{m\in \mathbb{R}~|~\forall y\in \mathrm{col}(A), \exists x\in \mathbb{R}^{m}~ \mathrm{s.t.}~ Ax=y~\mathrm{and}~ m\norm{x}\leq \norm{y}\}, 
$ see \cite{Grcar10}.
Matrix lower bounds satisfy the following properties, which are essential in our derivation of the proposed method.
\begin{lemma}\label{lem:MatrixLowerBounds}
Let $A\in \mathbb{R}^{n\times m}$ and $B\in \mathbb{R}^{m\times p}$, where $\mathrm{rank}(A)=n$ and $\mathrm{rank}(B)=m$ (full row rank). Then:
\begin{enumerate}
\item[(a)] $\norm{A}_{l} \mathbf{B}_{n}\subseteq A \mathbf{B}_{m}$ (follows from \cite[Lemma~2.3]{Grcar10}).

\item[(b)] $
{1}/{\norm{A^{\dagger}}}\leq \norm{A}_{l}$ (follows from \cite[Lemma~2.2]{Grcar10}).
\item [(c)]
$\norm{A}_{l}\norm{B}_{l}\leq \norm{AB}_{l}$ (follows from \cite[Lemma~4.4]{Grcar10}).

\end{enumerate}
\end{lemma}

The collection of  full-dimensional subsets of $\mathbb{R}^{n}$ that are affine transformations of closed unit balls is denoted by $\mathbb{A}_{n}$,
where, in this work, saying  $\Omega=c+G\mathbf{B}_{p}\in \mathbb{A}_{n}$ implies that $c\in \mathbb{R}^{n}$, $G\in \mathbb{R}^{n\times p}$, and $\mathrm{rank}(G)=n$. { Integration of single-valued functions presented herein is
always understood in the sense of Bochner}. Given a non-empty subset $X \subseteq \mathbb{R}^n$ and a measurable subset $S \subseteq \mathbb{R}$, $X^{S}$ denotes the set of Lebesgue measurable maps with domain  $S$ and  values in $X$.
%for almost every $t\in \intcc{a,b}$
%Integration is
%always understood in the sense of Bochner, an extension of Lebesgue integration \cite{HytonenVanNeervenVeraarWeis16}. 
%Almost every (where) is abbreviated as a.e.
Given a non-empty compact subset $W\subseteq \mathbb{R}^{m}$, a compact interval $\intcc{a,b}\subset \mathbb{R}$, and an integrable matrix-valued  function
$F: \intcc{a,b} \to \mathbb{R}^{n\times m}$, we define the set-valued integral
$
\int_{a}^{b} F(t) W \mathrm{d}t\defas \bigcup_{w \in W^{\intcc{a,b}}}\int_{a}^{b} F(t) w(t) \mathrm{d}t.
$ 
The Hausdorff distance $\mathfrak{d}(\Omega,\Gamma)$ of two non-empty
bounded subsets $\Omega, \Gamma \subseteq \mathbb{R}^n$
w.r.t.~$\| \cdot \|$ is defined as
$
\mathfrak{d}(\Omega,\Gamma)
\defas
\inf
\Set{ \varepsilon > 0}{%
\Omega \subseteq \Gamma + \varepsilon \mathbf{B}_{n},
\Gamma \subseteq \Omega + \varepsilon \mathbf{B}_{n}
}.
$
 The Hausdorff distance satisfies the triangle
inequality, in addition to the following  set
of  properties.
%%%%%%%%%%%%%%%%%%%%%%%%%%%%%%%%%%%%%%%%%%%%%%%%%%%%%%%%%%%%%%%%%%%%%%
% references for Hausdorff distance:
% ----------------------------------
% a) defined via containment in nbh\mathrm{d}s:
% Valentine64: nbhd = cBall; nonempty, bd, conv sets in rLNS
% DeBlasi76:   nbhd = oBall; metric for nonempty, bd, closed sets in BS
%
% b) defined via unsymmetric distances:
% Zeidler.i: in metric spaces
%
% c) both definitions:
% RockafellarWets09: nbhd = cBall; ... in R^n
% HuPapageorgiou97.i: nbhd = oBall; in metric space
%
% d) incorrect or just shitty:
% GoodmanORourke04: incorrect
% Hermes70: w.r.t. Euclidean distance and in R^n only
% LeStoicaAlamoCamachoDumur13: unreliable
% HuttenlocherKlandermanRucklidge93: finite sets in R^2 wrt Euclidean
%      metric (images)
%%%%%%%%%%%%%%%%%%%%%%%%%%%%%%%%%%%%%%%%%%%%%%%%%%%%%%%%%%%%%%%%%%%%%%
\begin{lemma}[Hausdorff distance]
\label{lem:HausdorffDistance}
Let $\Omega, \Omega', \Gamma, \Gamma' \subseteq \mathbb{R}^n$ be
non-empty and bounded, and let $A,B \in \mathbb{R}^{m \times n}$. Then, the following hold (see \cite[Lemma~A.2]{SerryReissig21}):
\begin{enumerate}
\item[(a)]
\label{lem:HausdorffDistance:1}
$\mathfrak{d}( \Omega + \Gamma, \Omega' + \Gamma' )
\le
\mathfrak{d}( \Omega, \Omega' ) + \mathfrak{d}( \Gamma, \Gamma' )$.
\item[(b)]
\label{lem:HausdorffDistance:2}
$
\mathfrak{d}( A \Omega, A \Gamma )
\le
\| A \| \mathfrak{d}( \Omega, \Gamma )
$.
\item[(c)]
\label{lem:HausdorffDistance:3}
$
\mathfrak{d}( A \Omega, B \Omega )
\le
\| A-B \| \| \Omega \|$ (implying $\mathfrak{d}(\Omega,0)\leq \norm{\Omega}$).

\item[(d)]
\label{lem:HausdorffDistance:4}
Let
$(\Omega_i)_{i\in I}$ and
$(\Gamma_i)_{i\in I}$ be families of non-empty subsets of
$\mathbb{R}^n$. Then, 
$
\mathfrak{d} \left(
\cup_{i \in I} \Omega_i,
\cup_{i \in I} \Gamma_i
\right)
\le
\sup_{i \in I} \mathfrak{d}( \Omega_i, \Gamma_i )
$.

%\item
%\label{lem:HausdorffDistance:4}
%$\| \Omega \| = \mathfrak{d}( \Omega, \{ 0 \} )$.
%\item
%\label{lem:HausdorffDistance:5}
%If $\Omega$ and $\Gamma$ are additionally closed, then
%$\mathfrak{d}( \Omega, \Gamma ) \le \varepsilon$
%iff
%$\Omega \subseteq \Gamma + \varepsilon \mathbf{B}$
%and
%$\Gamma \subseteq \Omega + \varepsilon \mathbf{B}$.
%\item
%\label{lem:HausdorffDistance:6}
%If $\mathfrak{d}( \Omega, \Gamma ) \le \varepsilon$, then
%$
%\mathfrak{d}( \Omega, \Gamma + \delta \mathbf{B} )
%\le
%\varepsilon + \delta
%$.
\end{enumerate}
\end{lemma}

\section{System Description and Problem Formulation}
\label{sec:SystemDescription}
In this paper, we  consider the LTI system
\begin{equation}\label{eq:LinearSystem}
\dot{x}(t)=A x(t)+u(t)
\end{equation}
 over the time interval $\intcc{0,T}$, where $x(t)\in \mathbb{R}^{n}$ is the system state, $u(t)\in \mathbb{R}^{n}$ is the input, and $A\in \mathbb{R}^{n\times n}$  is the system matrix. The initial value $x(0)$  and the input $u(t),~t\in\intcc{0,T}$  belong to known sets $X_{0}$ and $U$, respectively.  Let $T$,  $A$, $X_0$, and $U$ be fixed and assume that:
 \begin{enumerate}
  \item The time interval $\intcc{0, T}$ is compact and $T>0$.

\item $X_{0}=c_{x}+G_{x}\mathbf{B}_{p_x}\in \mathbb{A}_{n}$ and $U=c_{u}+G_{u}\mathbf{B}_{p_u}\in \mathbb{A}_{n}$,  where $c_{x}, c_{u}, G_{x}$, and $G_{u}$ are known. 
 \end{enumerate}

 Given an initial value $x(0)=x_{0}$ and an integrable input signal $v: \intcc{0,T}\rightarrow \mathbb{R}^{n}$,  the unique solution, $\varphi(\cdot, x_0,v)$, to system \eqref{eq:LinearSystem}, generated by  $x_0$ and $v(\cdot)$, on $\intcc{0,T}$ is given by \cite[Theorem~6.5.1,~p.~114]{Lukes82}
$$
\varphi(t, x_0,v)=\e^{tA}x_{0}+\int_{0}^{t}\e^{(t-s)A}v(s)\mathrm{d}s,~ t\in \intcc{0 , T}.
$$
Herein,  $\e^{(\cdot)A}$ (or $\exp((\cdot)A)$) is the matrix exponential function, which has the Taylor series expansion
$
\exp(t A)=\sum_{j=0}^{\infty}{(tA)^j}/{j!}.
$
Define 
$$
\mathcal{L}(t,k)\defas \sum_{j=0}^{k-1}\frac{(tA)^j}{j!},\quad \mathcal{T}(t,k)\defas \int_{0}^{t}\mathcal{L}(s,k)\mathrm{d}s=\sum_{j=0}^{k-1}\frac{t^{j+1}A^j}{(j+1)!},
$$
where $\mathcal{L}(t,k)$ is the truncated $(k-1)$th-order Taylor expansion of $\exp(t A)$ and $\mathcal{T}(t,k)$ is its definite integral. 
% The truncated $(k-1)$th order Taylor expansion of $\exp(t A)$ %with $k$ terms  
% is denoted  by $\mathcal{L}(t,k)$, i.e., 
% $$
% \mathcal{L}(t,k)= \sum_{j=0}^{k-1}\frac{(tA)^j}{j!}.
% $$
% Moreover, the definite integral of $\mathcal{L}(\cdot, k)$ over $\intcc{0,t}$, is denoted by $\mathcal{T}(t,k)$, that is,
% $$
% \mathcal{T}(t,k)= \int_{0}^{t}\mathcal{L}(s,k)\mathrm{d}s=\sum_{j=0}^{k-1}\frac{t^{j+1}A^j}{(j+1)!}.
% $$
%The functions $\exp((\cdot) A)$ and $\mathcal{L}(\cdot,\cdot)$ satisfy the following inequalities for all $t\in \mathbb{R}_{+}$ and $k\in \mathbb{N}$:
{
It is easy to verify that, for all $t\in \mathbb{R}_{+}$ and $k\in \mathbb{N}$, 
\begin{align}\label{eq:BoundonExpAndL}
\norm{\e^{tA}}&\leq \e^{t\norm{A}},~ \norm{\mathcal{L}(t,k)}\leq \e^{t\norm{A}},\\\label{eq:BoundonApproxError}
\norm{\e^{tA}-\mathcal{L}(t,k)}&\leq \theta(t\norm{A},k)\leq  \frac{(t\norm{A})^k}{k!}\e^{t\norm{A}},
\end{align}
where $\theta\colon \mathbb{R}_{+}\times \mathbb{N}\rightarrow \mathbb{R}_{+}$ is defined as
\begin{equation}\label{eq:Theta}
\theta(r,p)\defas \e^{r}-\sum_{j=0}^{p-1}\frac{r^{j}}{j!}=\sum_{j=p}^{\infty}\frac{r^{j}}{j!},~r\in \mathbb{R}_{+},~p\in \mathbb{N}.
\end{equation}
The function $\theta$ is  infinitely differentiable and monotonically increasing in its first argument, and monotonically decreasing in its second argument, with a greatest lower bound of zero.
}

Let $\mathcal{R}(t)$ denote the forward reachable set of system $\eqref{eq:LinearSystem}$ at time $t\in \intcc{0,T}$, with  initial values in $X_{0}$, and input signals with values in $U$. In other words, 
\begin{equation}\label{eq:R(t)}
\mathcal{R}(t)\defas\e^{t A}X_{0}+\int_{0}^{t}\e^{sA}U\mathrm{d}s,~t\in \intcc{0,T}.
\end{equation}
 The set $\exp(t A)X_{0}$ is referred to as the homogeneous reachable set at time $t$ and is denoted by $\mathcal{R}_{h}(t)$, and the set $\int_{0}^{t}\exp(sA)U\mathrm{d}s$ is referred to as the input reachable set at time $t$ and is denoted by $\mathcal{R}_{u}(t)$. 
 Furthermore, let $\intcc{a,b} \subseteq \intcc{0,T}$.  Then,
 $
 \mathcal{R}(\intcc{a,b})=\bigcup_{t\in \intcc{a,b}}\mathcal{R}(t)
 $
 is the reachable tube over the time interval $\intcc{a,b}$. In this paper, we aim to compute arbitrarily precise under-approximations of { $\mathcal{R}_{h}(T)$, $\mathcal{R}_{u}(T)$, $\mathcal{R}(T)$, and $\mathcal{R}(\intcc{0,T})$}, utilizing the approximations $\mathcal{L}$ and $\mathcal{T}$. 

\begin{remark}[Applications of under-approximations] \label{rem:Applications}
{ In this work,  we focus  on under-approximating  forward finite reachable sets for linear  systems with   uncertainties or constraints on the initial values and inputs. Under-approximations can be  beneficial  in control synthesis and verification applications. For example, let us consider the case when the input set $U$ corresponds to a disturbance set, and  $X_\mathrm{US}\subset \mathbb{R}^{n}$ be an unsafe set. If $
\mathcal{R}([0,T]) 
$, or an under-approximation of it, intersects with $X_\mathrm{US}$, then this indicates that the initial set $X_{0}$, for sure, does not satisfy safety specifications (see the framework of falsification, e.g., in \cite{BhatiaFrazzoli04}). %If the unsafe set is convex (or a finite union of convex sets), and the under-approximation of $
%\mathcal{R}([0,T]) 
%$ is a finite union of convex sets (which is the case in this work), then verifying the intersection is computationally feasible. 

Moreover, let $X_\mathrm{target}\subseteq \mathbb{R}^{n} $ be a target set, and let the set $U$ correspond to a control input set. Define the backward reachable set (see, e.g., \cite{Mitchell07}) 
$$
\mathcal{R}_{\mathrm{bw}} (X_\mathrm{target},T)\defas
\e^{T(-A)}X_\mathrm{target}+\e^{T(-A)}(-\mathcal{R}_{u}(T)).
$$
 If  an initial value of interest  $x_{0}\in \mathbb{R}^{n}$ belongs to $\mathcal{R}_{\mathrm{bw}} (X_\mathrm{target},T)$, or an under-approximation of it, then the existence of a control signal with values in $U$, driving $x_{0}$ to $X_\mathrm{target}$ in time $T$, is guaranteed, and such a control signal can be obtained by, e.g., solving an associated constrained optimal control problem. Note that  under-approximating $\mathcal{R}_{\mathrm{bw}} (X_\mathrm{target},T)$  requires under-approximating $\mathcal{R}_{u}(T)$, $\exp(T(-A))(-\mathcal{R}_{u}(T))$, and $\exp(T(-A))X_\mathrm{target}$. In this work, we address directly how to under-approximate $\mathcal{R}_{u}(T)$. Moreover, the tools in this work, and in particular Lemma \ref{lem:UnderApproximatingImageofLinearMap}, can be easily applied to under-approximate $\exp(T(-A))(-\mathcal{R}_{u}(T))$, and $\exp(T(-A))X_\mathrm{target}$.}  
\end{remark}

 \section{Proposed method}
 \label{sec:ProposedMethod}
In this section, we thoroughly derive the proposed method. The convergence guarantees of the method are discussed in Section \ref{sec:Convergence}.

 We start with the following theoretical recursive relation, which is  algorithmically similar to  efficient over-approximation methods in the literature \cite{GirardLeGuernicMaler06}, and is the basis of the proposed method of this work. 

 %Given $N\in \mathbb{N}$, define $\tau=T/N$, $\Gamma_{0}^{N}=X_{0}$,  
 %\begin{equation}\label{eq:TheoreticalMethod}
 %\Gamma_{i}^{N}=\e^{\tau A}\Gamma_{i}^{N}+\int_{0}^{\tau}\exp(sA)U\mathrm{d}s ,~i\in \intcc{1;N}
 %\end{equation}

\begin{lemma}\label{lem:TheoreticalRecursiveRelation}
 Given $N\in \mathbb{N}$, define $\tau=T/N$,   
\begin{subequations}
\label{eq:TheoreticalRecursiveRelation}
\begin{align}\label{eq:S}
S_{0}^{N}&=X_{0},&S_{i}^{N}&=\e^{\tau A}S_{i-1}^{N}, ~i\in \intcc{1;N},\\\label{eq:V}
V_{0}^{N}&=\mathcal{R}_{u}(\tau),&V_{i}^{N}&=\e^{\tau A}V_{i-1}^{N},~i\in \intcc{1;N},\\\label{eq:W}
W_{0}^{N}&=\{0\},& W_{i}^{N}&=W_{i-1}^{N}+V_{i-1}^{N},~i\in \intcc{1;N},\\
\label{eq:Gamma}
&&\Gamma_{i}^{N}&=S_{i}^{N}+W_{i}^{N},~i\in \intcc{0;N}.
\end{align}
\end{subequations}
Then,  for all $i\in \intcc{0;N},$
 $
 S_{i}^{N}= \mathcal{R}_{h}(i\tau), 
 $
 $
 W_{i}^{N}= \mathcal{R}_{u}(i\tau), 
 $ 
 and
 $
 \Gamma_{i}^{N}= \mathcal{R}(i\tau) 
 $. Moreover,   $\bigcup_{i=0}^{N}\Gamma_{i}^{N}\subseteq \mathcal{R}(\intcc{0,T})$.
 \end{lemma}
 \begin{proof}
By induction, $S_{i}^{N}=\exp(i\tau A)X_{0}=\mathcal{R}_{h}(i\tau),~i\in \intcc{0;N}$. According to \cite[Corollary~3.6.2,~p.~118]{Sontag98},  ${W}_{i}^{N}=\mathcal{R}_{u}(i\tau),~i\in \intcc{0;N}$. Therefore, $\Gamma_{i}^{N}=S_{i}^{N}+W_{i}^{N}=\mathcal{R}_{h}(i\tau)+\mathcal{R}_{u}(i\tau)=\mathcal{R}(i\tau),~i\in \intcc{0;N}$. The last claim follows from 
$\bigcup_{i\in \intcc{0;N}}\Gamma_{i}^{N}=\bigcup_{i\in \intcc{0;N}}\mathcal{R}(i\tau)\subseteq \bigcup_{t\in \intcc{0,T}}\mathcal{R}(t)= \mathcal{R}(\intcc{0,T})$.
 \end{proof}
{ The algorithm in Lemma \ref{lem:TheoreticalRecursiveRelation}, in theory, addresses the under-approximation problem of this work; however, this algorithm cannot be implemented exactly in general.} In the next sections, we address the challenges in implementing this theoretical algorithm and propose a practically implementable method, which is the main contribution of this work.

\subsection{Under-approximating  the image of the matrix exponential}
 The first obstacle in implementing the algorithm in Lemma \ref{lem:TheoreticalRecursiveRelation} is that recursive computations of the images of the  matrix exponential are required (see the definitions of $S_{i}^{N}$ and $V_{i}^{N}$ in Equations \eqref{eq:S} and \eqref{eq:V}, respectively), and exact computations of such images are not feasible in general as the exact value of the matrix exponential is generally not known. The following  technical lemma provides an insight into how to replace the aforementioned images with  under-approximations, where approximations of the matrix exponential can be utilized. 

\begin{lemma}\label{lem:UnderApproximatingImageofLinearMap}
Let $P,\tilde{P}\in \mathbb{R}^{n\times m}$ and $\Omega=c+G \mathbf{B}_{p} \in  \mathbb{A}_{m}$. Assume that  $P$ is of full row rank and that $\norm{(\tilde{P}-P)c} \leq \norm{PG}_l$. Then,
$
\tilde{P}(c+\alpha G \mathbf{B}_{p})\subseteq  P\Omega
$
for any $\alpha \in \intcc{0,\alpha_{m}(\Omega,P,\tilde{P})}$, where
\begin{equation}\label{eq:alpha_m}
\alpha_{m}(\Omega,P,\tilde{P})\defas \frac{\norm{PG}_{l}-\norm{(\tilde{P}-P)c}}{\norm{PG}_{l}+\norm{(\tilde{P}-P)G}}.
\end{equation}
\end{lemma}
\begin{proof}
Fix $\alpha\in \intcc{0,\alpha_{m}(\Omega,P,\tilde{P})}$. Note that $\alpha$ satisfies $\norm{(\tilde{P}-P)c}+\alpha\norm{(\tilde{P}-P)G }  \leq (1-\alpha) \norm{PG}_{l}$.  Define $F\defas \tilde{P}(c+\alpha G \mathbf{B}_{p})$. Then, using Lemma \ref{lem:MatrixLowerBounds}(a),
\begin{align*}
F
\subseteq& Pc+(\tilde{P}-P)c+\alpha(\tilde{P}-P)G \mathbf{B}_{p}+\alpha PG \mathbf{B}_{p} \\
\subseteq  & Pc+\norm{(\tilde{P}-P)c}\mathbf{B}_{n}+\alpha\norm{(\tilde{P}-P)G }\mathbf{B}_{n}+\alpha PG \mathbf{B}_{p}\\
= & Pc+\left(\norm{(\tilde{P}-P)c}+\alpha\norm{(\tilde{P}-P)G } \right)\mathbf{B}_{n}+\alpha PG \mathbf{B}_{p}\\ \subseteq&   Pc+ (1-\alpha) \norm{PG}_{l} \mathbf{B}_{n}+\alpha PG \mathbf{B}_{p} \\
\subseteq & Pc +(1-\alpha)P G \mathbf{B}_{p}+ \alpha PG \mathbf{B}_{p}=P(c+G \mathbf{B}_{p})=P \Omega.
\end{align*}
\end{proof}
{Lemma \ref{lem:UnderApproximatingImageofLinearMap} can be explained intuitively as follows. The set  $P(c+G\mathbf{B}_{p})$ needs to be under-approximated utilizing an approximation of $P$ (our approximation is $\tilde{P}$ in this case). The set $\tilde{P}(c+G\mathbf{B}_{p})$ resembles an approximation of $P(c+G\mathbf{B}_{p})$; however, it is not an under-approximation. By utilizing  estimates of the  approximation errors, the set $P(c+G\mathbf{B}_{p})$ is  deflated, by introducing the parameter $\alpha \in \intcc{0,\alpha_{m}(c+G\mathbf{B}_{p}, P,\tilde{P})}$, where $\alpha_{m}$ is defined by \eqref{eq:alpha_m}, making the set $P(c+\alpha G\mathbf{B}_{p})$ the desired under-approximation (see Figure \ref{fig:Schematic1}). }
\begin{figure}
    \centering
    \includegraphics[width=\columnwidth]{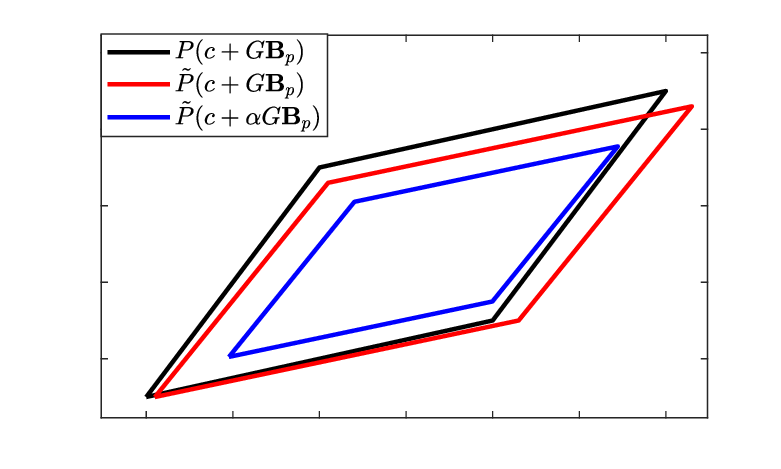}
    \caption{ Schematic representation of the result of Lemma \ref{lem:UnderApproximatingImageofLinearMap}: the set $P(c+G\mathbf{B}_{p})$ (black), its approximation $\tilde{P}(c+G\mathbf{B}_{p})$ (red), and the under-approximation $\tilde{P}(c+\alpha G\mathbf{B}_{p})$ (blue) (unit balls herein are w.r.t. the maximum norm), where $\alpha \in \intcc{0,\alpha_{m}(c+G\mathbf{B}_{p}, P,\tilde{P})}$ and $\alpha_{m}$ is defined by \eqref{eq:alpha_m} .}
    \label{fig:Schematic1}
\end{figure}

Utilizing Lemma \ref{lem:UnderApproximatingImageofLinearMap}, we can obtain under-approximations of the images of the matrix exponential using truncated Taylor expansions as shown in Corollary \ref{cor:UnderApproximatingRX} below. %Hence, the first obstacle  in the process of deriving a practical variant of the algorithm in Lemma \ref{lem:TheoreticalRecursiveRelation} is overcome. 
{Before doing so, we first introduce the deflation function $\lambda$, which  plays a major role in  determining the extent of which the considered sets in our analysis need to be ``shrunk''  in order to obtain under-approximations.} 
\begin{definition}[Deflation coefficient]
Given $t\in \mathbb{R}_{+}$, $\Omega=c+G\mathbf{B}_{p}\in \mathbb{A}_{n}$, and $k\in \mathbb{N}$, define the deflation parameter 
\begin{equation}\label{eq:lambda}
\lambda(t,\Omega,k)\defas  \frac{ 1-\e^{t\norm{A}}\theta(t\norm{A},k)\norm{G^{\dagger}}\norm{c}}{ 1+\e^{t\norm{A}}\theta(t\norm{A},k)\norm{G^{\dagger}}\norm{G}}.
\end{equation}
\end{definition}

\begin{corollary}\label{cor:UnderApproximatingRX}
Given  $\Omega=c+G\mathbf{B}_{p}\in \mathbb{A}_{n}$, and $t\in \mathbb{R}_{+}$, then for all $k\in \mathbb{N}$ such that $\lambda(t,\Omega,k) \geq 0$,
we have
$
\mathcal{L}(t,k)[c+\lambda(t,\Omega,k) G]\subseteq \exp(t A)\Omega.
$
\end{corollary}  

\begin{proof}
The corollary follows from Lemma \ref{lem:UnderApproximatingImageofLinearMap} (with $\tilde{P}=\mathcal{L}(t,k)$ and $P=\exp(t A)$) by verifying that
\begin{align*}
\alpha_{m}(\Omega,\e^{tA},\mathcal{L}(t,k))&=\frac{\norm{\e^{tA}G}_{l}-\norm{(\mathcal{L}(t,k)-\e^{tA})c}}{\norm{\e^{tA}G}_{l}+\norm{(\mathcal{L}(t,k)-\e^{tA})G}}.\\
&\geq \frac{ (\norm{G^{\dagger}}\norm{\e^{-tA}})^{-1}-\norm{(\mathcal{L}(t,k)-\e^{tA})c}}{ (\norm{G^{\dagger}}\norm{\e^{-tA}})^{-1}+\norm{(\mathcal{L}(t,k)-\e^{tA})G}}\\
& \geq \frac{ 1 -\norm{G^{\dagger}}\e^{t\norm{A}}\norm{(\mathcal{L}(t,k)-\e^{tA})c}}{ 1+\norm{G^{\dagger}}\e^{t\norm{A}}\norm{(\mathcal{L}(t,k)-\e^{tA})G}}\\
& \geq  \frac{ 1 -\e^{t\norm{A}}\norm{G^{\dagger}}\theta(t\norm{A},k)\norm{c}}{ 1+\e^{t\norm{A}}\norm{G^{\dagger}}\theta(t\norm{A},k)\norm{G}}\\
&=\lambda(t,\Omega,k)\geq 0,
\end{align*}
where we have used Lemma \ref{lem:MatrixLowerBounds}(b),(c).  
%$\alpha_{m}(\Omega,\e^{tA},\mathcal{L}(t,k))\geq \lambda(t,\Omega,k)\geq 0$
\end{proof}
{
\subsection{Invertible truncated Taylor series of the matrix exponential}
\label{sec:InvertibleTaylorSeries}
In the algorithm presented in Lemma \ref{lem:TheoreticalRecursiveRelation}, images of the matrix exponential (the sets $S_{i}^{N}$ and $V_{i}^{N}$) are computed recursively. As we aim to adopt Corollary \ref{cor:UnderApproximatingRX} to replace these exact images with under-approximations, it is important that
    each under-approximation is in the class $\mathbb{A}_{n}$ (cf. $\Omega$ in Corollary \ref{cor:UnderApproximatingRX}). 
    %(note how we require in Corollary \ref{cor:UnderApproximatingRX} that the set  $\Omega$ (to be linearly transformed) is in  $\mathbb{A}_{n}$).
 %This will necessitate %, as can be deduced from Corollary \ref{cor:UnderApproximatingRX} and Lemma \ref{lem:TheoreticalRecursiveRelation},
This necessitates that:
 \begin{enumerate}
     \item the set $X_{0}$ is in $\mathbb{A}_{n}$,  which holds by assumption;
     \item the under-approximation of $\mathcal{R}_{u}(\tau)$, which we derive and discuss  in Sections \ref{subsec:UnderApproximatingInputReachSet} and \ref{subsec:InvertibilityOfIntegral}, is in $\mathbb{A}_{n}$; and
     \item in each iteration of the method, when under-approximating $S_{i}^{N}$ and $V_{i}^{N}$ using Corollary \ref{cor:UnderApproximatingRX},  the values of the function $\lambda$ are positive, and  the values of $\mathcal{L}$ are invertible.
 \end{enumerate}
 The positivity of $\lambda$ can be imposed by setting the number of Taylor  terms in $\mathcal{L}$ to be sufficiently large (see Section \ref{subsec:NumTaylorTerms} for details). 
 %The positivity of $\lambda$ can be imposed by setting the number of Taylor series terms associated with $\mathcal{L}$ to be sufficiently large and we will discuss that thoroughly in Section \ref{subsec:NumTaylorTerms}. 
 
 Unfortunately, the invertibility of truncated Taylor expansions of the matrix exponential is not guaranteed in general. For example,  if $A=-\id$,  
 then $\mathcal{L}(1,2)=0$ (not invertible). 
 The following lemma provides an efficient way to check the invertibility of $\mathcal{L}(t,k)$. }
 
 %A repetitive  check of invertibility of the implemented  Taylor approximations may increase the computational complexity of the proposed method, especially in higher dimensions. To overcome this issue, we compute a lower bound on the number of Taylor series terms that guarantees  invertibility as seen in the lemma below: 
%{ The above Lemma actually holds for any induced norm \cite[Theorem~4.7]{Friedland82}, so the 2-norm of $A$ in the definition of $k_{\min}$ below can be replaced with the spectral radius of $A$.   }
{
\begin{lemma}\label{Lem:Kmin}
Let $t\in \mathbb{R}_{+}$ and define 
\begin{equation}\label{eq:Kmin}
k_{\min}(t)\defas \min \Set{k\in \mathbb{N}}{\theta(t\rho(A),k)\e^{t\rho(A)}<1}.
\end{equation}
Then, $\mathcal{L}(t,k)$ is invertible for all $k\in \intco{k_{\min}(t);\infty}$.
\end{lemma}

\begin{proof}
First, we note that $k_{\min}(t)$ is well-defined (finite) as $\lim_{k\rightarrow\infty}\theta(r,k)=0$ for any $r\in \mathbb{R}_{+}$ and that the inequality $\theta(t\rho(A),k)\exp(t\rho(A))<1$ holds for any $k\in \intco{k_{\min}(t);\infty}$ as $\theta$ is monotonically decreasing with respect to its second argument. Fix $k\in \intco{k_{\min}(t);\infty}$.  The continuity of $\theta(t(\cdot),k)\exp(t(\cdot))$  implies the existence of some $\varepsilon>0$ such that 
$
\theta(t(\rho(A)+\varepsilon),k)\exp(t(\rho(A)+\varepsilon))<1.
$
Using Lemma \ref{lem:SpectralRadius}, there exists an induced matrix norm $\norm{\cdot}_{\varepsilon}$ such that $\norm{A}_\varepsilon\leq \rho(A)+\varepsilon$. Hence, using estimates \eqref{eq:BoundonExpAndL}
 and \eqref{eq:BoundonApproxError} and the increasing monotonicity of $\theta(t(\cdot),k)\exp(t(\cdot))$,
 \begin{align*}
 \norm{\e^{tA}-\mathcal{L}(t,k)}_\varepsilon \norm{(\e^{tA})^{-1}}_\varepsilon&= \norm{\e^{tA}-\mathcal{L}(t,k)}_\varepsilon \norm{\e^{-tA}}_\varepsilon\\
&\leq
\theta(t\norm{A}_\varepsilon,k)\e^{t\norm{A}_\varepsilon}\\
&\leq\theta(t(\rho(A)+\varepsilon),k)\e^{t(\rho(A)+\varepsilon)}<1.
\end{align*}
Finally, using Lemma \ref{Lem:ContinuityOfMatrixInverse}, $\mathcal{L}(t,k)$ is invertible.
\end{proof}
}

\subsection{Under-approximating input reachable sets}
\label{subsec:UnderApproximatingInputReachSet}
The second main issue with implementing the Algorithm in Lemma \ref{lem:TheoreticalRecursiveRelation} is the requirement to evaluate the input reachable set $\mathcal{R}_{u}(\tau)$ exactly, which is generally not possible. The following lemma is the starting point to address this issue.

\begin{lemma}
\label{lem:UnderApproximatingSetValuedIntegral}
Let  $I=\intcc{a,b}$ be a compact interval,  $P, \tilde{P}: I\rightarrow \mathbb{R}^{n\times m}$  be  continuous matrix-valued functions,  and $\Omega=c+G\mathbf{B}_{p} \in \mathbb{A}_{m}$.  Assume that $\mathrm{rank}(P(s))=n$ for all $s\in I$. Moreover, assume $\sup_{s\in I}\norm{(\tilde{P}(s)-P(s))c}\leq \inf_{s\in I}\norm{P(s)G}_{l} $. Then,
$
\int_{I}\tilde{P}(s)\mathrm{d}s (c+\alpha G \mathbf{B}_{p})\subseteq  \int_{I}P(s) \Omega \mathrm{d}s
$
for any $\alpha \in \intcc{0,\gamma_{m}(\Omega,t_{1},t_{2},P,\tilde{P})}$, where
\begin{equation}\label{eq:gamma_m}
\gamma_{m}(\Omega,I,P,\tilde{P})\defas \frac{\inf_{s\in I}\norm{P(s)G}_{l}-\sup_{s\in I}\norm{P_{d}(s)c}}{\inf_{s\in I}\norm{P(s)G}_{l}+\sup_{s\in I}\norm{P_{d}(s)G}},
\end{equation}
and $P_{d}(s)\defas\tilde{P}(s)-P(s),~s\in I$.
\end{lemma}
\begin{proof}
Fix $\alpha \in \intcc{0,\gamma_{m}(\Omega,I,P,\tilde{P})}$ and define 
$$
\mu=\sup_{s\in I}\norm{(\tilde{P}(s)-P(s))c} +\alpha\sup_{s\in I}\norm{(\tilde{P}(s)-P(s))G}.
$$
Then,  $\alpha$ satisfies $ \mu\leq (1-\alpha)\inf_{s\in I}\norm{P(s)G}_{l}.$
Let $y\in \int_{I}\tilde{P}(s)\mathrm{d}s(c+\alpha G \mathbf{B}_{p})$. Then, there exists $x\in \mathbf{B}_{p}$ such that $y=\int_{I}\tilde{P}(s)\mathrm{d}s(c+\alpha G x)=\int_{I}f(s)\mathrm{d}s$, where $f(s)=\tilde{P}(s)(c+\alpha G x)$.  Therefore, using Lemma \ref{lem:MatrixLowerBounds}(a), we have, for all $s\in I$, 
\begin{align*}
f(s)=&  P(s)c+(\tilde{P}(s)- P(s))c\\
&+\alpha(\tilde{P}(s)-P(s))G x+\alpha P(s) G x\\
 \subseteq &   P(s)c+\mu\mathbf{B}_{n}+\alpha P(s) G \mathbf{B}_{p}\\
 \subseteq &  P(s)c+ (1-\alpha) \inf_{z\in I}\norm{P(z)G}_{l}\mathbf{B}_{n}+\alpha P(s) G \mathbf{B}_{p}\\
 \subseteq & P(s)c+ (1-\alpha) \norm{P(s)G}_{l}\mathbf{B}_{n}+\alpha P(s) G \mathbf{B}_{p}\\
 \subseteq & P(s)c+ (1-\alpha) P(s)G\mathbf{B}_{p}+\alpha P(s) G \mathbf{B}_{p}=P(s)\Omega.
\end{align*}
Using \cite[Theorem~8.2.10,~p.~316]{AubinFrankowska00}, there exists  $g \in \Omega^{I}$ such that $f(s)=P(s)g(s)$ for almost all $s\in I.$ Hence,  $y=\int_{I}P(s)g(s)\mathrm{d}s \in \int_{I}P(s)\Omega \mathrm{d}s$,
which completes the proof.
\end{proof}
{ The logic of Lemma \ref{lem:UnderApproximatingSetValuedIntegral} is  similar to that of Lemma \ref{lem:UnderApproximatingImageofLinearMap}, where the set-valued integral $\int_{I}P(s)(c+G\mathbf{B}_{p})\mathrm{d}s$  is under-approximated by a deflated version of the set $\int_{I}\tilde{P}(s)\mathrm{d}s (c+G\mathbf{B}_{p})$, utilizing the approximating matrix function $\tilde{P}$. Note that the set representation of $\int_{I}\tilde{P}(s)\mathrm{d}s (c+G\mathbf{B}_{p})$ (linear transformation of  $(c+G\mathbf{B}_{p})$ under $\int_{I}\tilde{P}(s)\mathrm{d}s$)  is  simpler than that of the original set-valued integral, making it more appealing in set-valued computations. The deflation herein is obtained based on estimates of  approximation errors as seen in the definition of $\gamma_{m}$ given by Equation \eqref{eq:gamma_m}.
} 

Lemma \ref{lem:UnderApproximatingSetValuedIntegral} provides a sufficient tool to under-approximate the input reachable set as seen in the following corollary.
\begin{corollary}\label{cor:UnderApproximatingRU}
Given  $\Omega=c+G\mathbf{B}_{p}\in \mathbb{A}_{n}$, and $t\in \mathbb{R}_{+} $, then for all $k\in \mathbb{N}$ such that $\lambda(t,\Omega,k) \geq 0$,
we have
$
\mathcal{T}(t,k)(c+\lambda(t,\Omega,k) G \mathbf{B}_{p})\subseteq \int_{0}^{t}\exp(sA)\Omega\mathrm{d}s.
$
\end{corollary}  

\begin{proof}
This follows by verifying that $\gamma_{m}(\Omega,I,\e^{(\cdot) A},\mathcal{L}(\cdot,k)) \geq \lambda(t,\Omega,k)\geq 0 $, where the detailed estimates, which are similar to those presented in the proof of Corollary \ref{cor:UnderApproximatingRX}, are omitted  for brevity. 
% Let $I=\intcc{0,t}$, $\tilde{P}(s)=\mathcal{L}(s,k)$, $P(s)=\exp(sA)$ for $s\in I$. The corollary follows by verifying
% \begin{flalign*}
% &\gamma_{m}(\Omega,I,\e^{(\cdot) A},\mathcal{L}(\cdot,k)) =\\
% &\frac{\inf_{s\in I}\norm{\e^{sA}G}_{l}-\sup_{s\in I}\norm{(\mathcal{L}(s,k)-\e^{sA})c}}{\inf_{s\in I}\norm{\e^{sA}G}_{l}+\sup_{s\in I}\norm{(\mathcal{L}(s,k)-\e^{sA})G}}&\\&
% \geq  \frac{\inf_{s\in I}(\norm{G^\dagger}\norm{\e^{-sA}})^{-1}-\sup_{s\in I}\norm{(\mathcal{L}(s,k)-\e^{sA})c}}{\inf_{s\in I}(\norm{G^\dagger}\norm{\e^{-sA}})^{-1}+\sup_{s\in I}\norm{(\mathcal{L}(s,k)-\e^{sA})G}}&\\&
% \geq  \frac{(\sup_{s\in I}\norm{G^\dagger}\norm{\e^{-sA}})^{-1}-\sup_{s\in I}\norm{(\mathcal{L}(s,k)-\e^{sA})c}}{(\sup_{s\in I}\norm{G^\dagger}\norm{\e^{-sA}})^{-1}+\sup_{s\in I}\norm{(\mathcal{L}(s,k)-\e^{sA})G}}&\\&
% \geq  \frac{1-\sup_{s\in I}\norm{G^\dagger}\norm{\e^{-sA}}\sup_{s\in I}\norm{(\mathcal{L}(s,k)-\e^{sA})c}}{1+\sup_{s\in I}\norm{G^\dagger}\norm{\e^{-sA}}\sup_{s\in I}\norm{(\mathcal{L}(s,k)-\e^{sA})G}}&\\&
% \geq  \frac{1-\sup_{s\in I}\norm{G^\dagger}\e^{s\norm{A}}\sup_{s\in I}\frac{(s\norm{A})^{k}}{k!}\e^{s\norm{A}}\norm{c}}{1+\sup_{s\in I}\norm{G^\dagger}\e^{s\norm{A}}\sup_{s\in I}\frac{(s\norm{A})^{k}}{k!}\e^{s\norm{A}}\norm{G}}&\\
% & \geq \lambda(t,\Omega,k)\geq 0,
% \end{flalign*}
% where we have used Lemma \ref{lem:MatrixLowerBounds}(b),(c). 
\end{proof}  
{
\subsection{Invertibility of the integral of the matrix exponential}
\label{subsec:InvertibilityOfIntegral}
Corollary \ref{cor:UnderApproximatingRU} establishes how the input reachable set $\mathcal{R}_{u}(\tau)$ required in the algorithm in Lemma \ref{lem:TheoreticalRecursiveRelation} can be replaced by an under-approximation that is an affine transformation of a unit ball (not necessarily full-dimensional) as $U\in \mathbb{A}_{n}$ by assumption. A blind implementation of the corollary may, however,  lead to another issue if the under-approximation of $\mathcal{R}_{u}(\tau)$ is not a full dimensional set (i.e., not in $\mathbb{A}_{n}$). As we mentioned in Section \ref{subsec:NumTaylorTerms}, we need the under-approximation of $\mathcal{R}_{u}(\tau)$ to be in  $\mathbb{A}_{n}$. %According to the theoretical algorithm in Lemma   \ref{lem:TheoreticalRecursiveRelation}, the input reachable set $\mathcal{R}_{u}(\tau)$ will be linearly transformed iteratively using the matrix exponential $\exp(\tau A)$ to obtain the sets $V_{i}^{N},~i\in \intcc{1;N}$. 
%As we will be using Corollary %\ref{cor:UnderApproximatingRX} to replace %the images of the matrix exponential, we %necessitate that our under-approximation %of $\mathcal{R}_{u}(\tau)$ is in %$\mathbb{A}_{n}$,
This, as can be seen from Corollary \ref{cor:UnderApproximatingRU}, is  attained if the utilized value of $\lambda$ is positive and  the value of the approximating function $\mathcal{T}$  is invertible. The first requirement can be fulfilled by setting the number of Taylor series terms of the approximation $\mathcal{T}$ to be large and we elaborate regarding that in Section \ref{subsec:NumTaylorTerms}. What is left is to ensure the invertibility of the value of  $\mathcal{T}$, which we can ensure if the integral of the matrix exponential itself is invertible.

To further investigate the invertibility requirement,  let us first  introduce  the set $\mathbb{I}$ of positive $t$ values such that $\int_{0}^{t}\exp(sA)\mathrm{d}s$ is invertible, i.e.,
 $$
 \mathbb{I}\defas\Set{t\in \mathbb{R}_{+}}{\int_{0}^{t}\exp(sA)\mathrm{d}s~\textrm{is invertible}}.
 $$
%The proposed method presented in Theorem \ref{thm:ProposedMethod} assumes the invertibility of  $\int_{0}^{\tau}\exp(sA)\mathrm{d}s$, which, in fact, is nonrestrictive. To see this, 
This set can be deduced exactly from the eigenvalues of $A$ as shown in the lemma below (see, e.g.,  \cite[Lemma~3.4.1,~p.~100]{Sontag98}).
\begin{lemma}
\label{lem:IntegralInvertible}
Let $t\in \mathbb{R}_{+}\setminus\{0\}$, then  $\int_{0}^{t}\exp(sA)\mathrm{d}s$ is invertible iff $2\pi z \mathbf{i} /t$ is not an eigenvalue of $A$ for any $z\in \mathbb{Z}\setminus\{0\}$, where $\mathbf{i}=\sqrt{-1}$.
\end{lemma}
We can see from Lemma \ref{lem:IntegralInvertible} that  the invertibility of $\int_{0}^{t}\exp(sA)\mathrm{d}s$ fails at only a countable number of values of $t$. Hence, in practice, the invertibility is likely always fulfilled. Furthermore, we can use Lemma \ref{lem:IntegralInvertible} to show that $\int_{0}^{t}\exp(sA)\mathrm{d}s$ is always invertible for $t$ sufficiently small, but nonzero (which is the case as we use it with $t=\tau$).  For a given system matrix $A$, the following lemma derives a fixed open interval on which the matrix exponential is guaranteed to be invertible.
%Furthermore, we can even  derive from Lemma \ref{lem:IntegralInvertible} an open interval, with zero as a left limit point, on which the integral of the matrix exponential is guaranteed to be invertible as shown  below:
\begin{lemma}\label{lem:InvertibilityInterval}
Let $E=\{ m\in \mathbb{R}_{+}~|~m=\abs{\lambda},~Av=\lambda  v,~\lambda\in \mathbb{C},~v\in \mathbb{C}^{n},~\Re(\lambda)=0\}\setminus \{0\}$ be the set of absolute values of the purely imaginary (nonzero) eigenvalues of $A$.  If $E$ is non-empty,  set $t_{\max}=2\pi/\max{E}$, otherwise  $t_{\max}=\infty$. Then, $\int_{0}^{t}\exp(sA)\mathrm{d}s$ is invertible for all $t\in \intoo{0,t_{\max}}$.
\end{lemma}
{ As can be seen from Lemma \ref{lem:InvertibilityInterval}, we can always find, based on the eigenvalues of the system matrix $A$, which is \textit{fixed} and known for a given system, a non-empty open interval with zero as a left limit point on which the integral of the matrix exponential is guaranteed to be invertible.}
In our proposed method, %which is based on the theoretical algorithm in Lemma \ref{lem:TheoreticalRecursiveRelation}, 
we are interested in an under-approximation of $\mathcal{R}_{u}(\tau)$, where $\tau$ is typically small as it corresponds to the time step size of the method ($\tau=T/N$). Therefore, the invertibility of $\int_{0}^{\tau}\exp(sA)\mathrm{d}s$ can always be fulfilled if we set the time discretization parameter $N$ to be sufficiently large (but still  finite), depending on the eigenvalues of $A$.

If the invertibility of the integral of the matrix exponential is fulfilled ($t\in \mathbb{I}$), we can obtain a finite $k\in \mathbb{N}$ such that the truncated Taylor series of  $\int_{0}^{t}\exp(sA)\mathrm{d}s$, $\mathcal{T}(t,k)$, is invertible.   The existence of a finite $k$, such that  $\mathcal{T}(t,k)$ is invertible, follows from the invertibility of $\int_{0}^{t}\exp(sA)\mathrm{d}s$ ($t\in \mathbb{I}$), the fact that $\lim_{k\rightarrow \infty}\mathcal{T}(t,k)= \int_{0}^{t}\exp(sA)\mathrm{d}s$, and the continuity of the matrix inverse (Lemma \ref{Lem:ContinuityOfMatrixInverse}).
}

\subsection{Determining the number of Taylor series terms}
\label{subsec:NumTaylorTerms}
Next, we aim to determine the number of Taylor series terms used in the approximations  $\mathcal{L}$ and $\mathcal{T}$. First, note that the under-approximations introduced in Corollaries \ref{cor:UnderApproximatingRX} and \ref{cor:UnderApproximatingRU} depend on the deflation function $\lambda$.  In this work, we propose  simple criteria, that  determine the number of Taylor series terms, which  aim to  maximize the values of $\lambda$, while incorporating the full-dimensionality considerations discussed in Sections \ref{subsec:NumTaylorTerms} and \ref{subsec:InvertibilityOfIntegral}. Notice that for any given $t\in \mathbb{R}_{+}$ and $\Omega=c+G\mathbf{B}_{p}\in \mathbb{A}_{n}$,  the function 
$\lambda(t,\Omega,\cdot)$ is  bounded above with the least upper bound of
$1$.  Our goal is to choose the minimum  value of $k$ such that the deflation coefficient is larger than some design parameter $\epsilon \in \intco{0,1}$, while ensuring the invertibility of the approximations of the matrix exponential and its integral. Therefore, we introduce the following  parameters associated  with the number of Taylor series terms used in under-approximating homogeneous and input reachable sets.

\begin{definition}
Given  $\Omega=c+G\mathbf{B}_{p}\in \mathbb{A}_{n}$, $t\in \mathbb{R}_{+}$, $\bar{t}\in \mathbb{I}$, and $\epsilon\in \intco{0,1}$.  $\kappa(t,\Omega,\epsilon)$ is defined as 
 \begin{equation}\label{eq:Kappa}
\kappa(t,\Omega,\epsilon)\defas\min \left( \Set{k\in \intco{k_{\min}(t);\infty}}{\lambda(t,\Omega,k)> \epsilon}\cap \intco{2;\infty}\right).
\end{equation}
Moreover, $\eta(\bar{t},\Omega,{\epsilon})$ is defined as
 \begin{equation}\label{eq:Eta}
\eta(\bar{t},\Omega,{\epsilon})\defas\min \{k\in \mathbb{N}~|~\lambda(\bar{t},\Omega,k)> {\epsilon},\mathcal{T}(\bar{t},k)~\textrm{is invertible}\}.  
\end{equation} 
\end{definition}
As shown in  the proof of Lemma \ref{Lem:Kmin},  $k_{\min}(t)$ is well-defined. Since 
$\lambda(t,\Omega,k)\rightarrow 1$, as $k\rightarrow \infty$,   $\kappa(t,\Omega,\epsilon)$ is also  well-defined for any $\epsilon\in \intco{0,1}$. Note that, the lower bound  of 2 used in the definition of $\kappa(t,\Omega,{\epsilon})$  is  of importance only when deducing the convergence guarantees in Section \ref{sec:Convergence}. The well-definiteness  of  $\eta(\bar{t},\Omega, \epsilon)$ follows from the fact that $\lim_{k\rightarrow \infty}\lambda(\bar{t},\Omega,k)= 1$ and the invertibility argument at the end of Section \ref{subsec:InvertibilityOfIntegral}.

\subsection{Under-approximations of reachable sets and tubes}

The previous sections have established the  tools necessary for the proposed method. Next, we introduce the  operators, $\mathcal{H}$ and $\mathcal{I}$, which are designed based on Corollaries \ref{cor:UnderApproximatingRX} and \ref{cor:UnderApproximatingRU}, in addition to the criteria introduced in \eqref{eq:Kappa} and \eqref{eq:Eta}, to obtain full-dimensional under-approximations of homogeneous and input reachable sets, using approximations of the matrix exponential and its integral.
\begin{definition}
We define the homogeneous and input under-approximation operators $\mathcal{H}\colon \mathbb{R}_{+}\times \mathbb{A}_{n}\times \intoc{0,1}\rightarrow \mathbb{A}_{n}$ and $\mathcal{I}\colon \mathbb{I}\times \mathbb{A}_{n}\times \intoc{0,1}\rightarrow \mathbb{A}_{n}$ as follows:
\begin{equation}\label{eq:H}
\mathcal{H}(t,\Omega,\epsilon)\defas\mathcal{L}(t,\kappa(t,\Omega,\epsilon))[c+\lambda(t,\Omega, \kappa(t,\Omega,\epsilon))G \mathbf{B}_{p}],
\end{equation}
\begin{equation}\label{eq:I}
\mathcal{I}(\bar{t},\Omega,\epsilon)\defas\mathcal{T}(\bar{t},\eta(\bar{t},\Omega,\epsilon))[c+ \lambda(\bar{t},\Omega, \eta(\bar{t},\Omega,{\epsilon}))G \mathbf{B}_{p}],
\end{equation}
 where $t\in \mathbb{R}_{+}$, $\bar{t}\in \mathbb{I}$, $\Omega=c+G \mathbf{B}_{p}\in \mathbb{A}_{n}$, and $\epsilon\in \intco{0,1}$.
\end{definition}

Now, we are ready to introduce the proposed method.

\begin{theorem}\label{thm:ProposedMethod}
Given $N\in \mathbb{N}$ and $\epsilon_{h},{\epsilon}_{u}\in \intco{0,1}$, define $\tau=T/N$, and assume $\tau \in  \mathbb{I}$. Moreover, define   
\begin{subequations}
\label{eq:ProposedMethod}
\begin{align}\label{eq:MathcalS}
\mathcal{S}_{0}^{N}&=X_{0},&\mathcal{S}_{i}^{N}&=\mathcal{H}(\tau, S_{i-1}^{N},\epsilon_{h}),~i\in \intcc{1;N},\\\label{eq:MathcalV}
\mathcal{V}_{0}^{N}&=\mathcal{I}(\tau,U,{\epsilon}_{u}),&\mathcal{V}_{i}^{N}&=\mathcal{H}(\tau,\mathcal{V}_{i-1}^{N},\epsilon_{h}),~i\in \intcc{1;N},\\\label{eq:MathcalW}
 \mathcal{W}_{0}^{N}&=\{0\},&\mathcal{W}_{i}^{N}&=\mathcal{W}_{i-1}^{N}+\mathcal{V}_{i-1}^{N},~i\in \intcc{1;N},\\\label{eq:LambdaUA}
&&\Lambda_{i}^{N}&=\mathcal{S}_{i}^{N}+\mathcal{W}_{i}^{N},~i\in \intcc{0;N}.
\end{align}
\end{subequations}
%Then, $\mathcal{S}_{i}^{N}$, $\mathcal{V}_{i}^{N}$, $\mathcal{W}_{i}^{N}$, and $\Lambda_{i}^{N}$ are nonempty for all $i\in \intcc{0;N}$. Moreover, 
Recall the definitions of $\{S_{i}^{N}\}_{i=0}^{N}$, $\{V_{i}^{N}\}_{i=0}^{N}$, $\{W_{i}^{N}\}_{i=0}^{N}$ and $\{\Gamma_{i}^{N}\}_{i=0}^{N}$ given in Lemma \ref{lem:TheoreticalRecursiveRelation}. Then,  $\mathcal{S}_{i}^{N}\subseteq S_{i}^{N}=\mathcal{R}_{h}(i\tau)$,  $\mathcal{V}_{i}^{N}\subseteq V_{i}^{N}$, $\mathcal{W}_{i}^{N}\subseteq W_{i}^{N}=\mathcal{R}_{u}(i\tau)$, and 
$
\Lambda_{i}^{N} \subseteq \Gamma_{i}^{N}= \mathcal{R}(i\tau)
$
for all $i \in \intcc{0;N}$. Furthermore, $\bigcup_{i=0}^{N} \Lambda_{i}^{N}\subseteq \mathcal{R}(\intcc{0,T})$.
\end{theorem}
\begin{proof}
%First, we prove the non-emptiness of the sets computed using  the proposed method.  Using the definition of $\mathcal{I}$ in Equation \eqref{eq:I},
%$
%\mathcal{I}(\tau,U,{\epsilon}_{u})=\mathcal{T}(\tau,\eta(\tau,U,{\epsilon}_{u}))[c_{u}+ \lambda(\tau,U, \eta(\tau,U,{\epsilon}_{u}))G_{u} \mathbf{B}_{p}]=\bar{c}+\bar{G} \mathbf{B}_{p_{u}},
%$
%where $\bar{c}=\mathcal{T}(\tau,\eta(t,U,{\epsilon}_{u}))c_{u}$ and $\bar{G}=\lambda(\tau,U, \eta(\tau,U,{\epsilon}_{u}))\mathcal{T}(\tau,\eta(\tau,U,{\epsilon}_{u}))G_{u}$. By the definition of $\eta$ in Equation \eqref{eq:Eta}, $\mathcal{T}(\tau,\eta(\tau,U,{\epsilon}_{u}))$ is invertible and $\lambda(\tau,U, \eta(\tau,U,{\epsilon}_{u}))>0$. Hence, $\bar{G}$ is of full row rank. Using Lemma \ref{lem:WellDefiniteness}, $\mathcal{S}_{i}^{N}$ and $\mathcal{V}_{i}^{N}$ are nonempty for all $i\in \intcc{0;N}$. Note that each  $\mathcal{W}_{i}^{N},~i\in \intcc{0;N}$ can be written as $\sum_{j=1}^{i-1}\mathcal{V}_{j}^{N}$ (here, $\sum_{j=1}^{0}(\cdot)=\{0\}$), hence, the non-emptiness of $\mathcal{V}_{i}^{N},~\in \intcc{0;N}$, implies the non-emptiness of $\mathcal{W}_{i}^{N},~i\in \intcc{0;N}$. Finally, the non-emptiness of $\Lambda_{i}^{N},~i\in \intcc{0;N}$ follows from the fact that  $\Lambda_{i}^{N}=\mathcal{S}_{i}^{N}+\mathcal{W}_{i}^{n},~i\in \intcc{0;N}$. Next, we prove the second claim of the theorem.  
We have $S_{0}^{N}=\mathcal{S}_{0}^{N}=X_{0}$. Assume that $\mathcal{S}_{i}^{N}\subseteq S_{i}^{n}$ for some $i\in \intcc{0;N-1}$, then, using Corollary \ref{cor:UnderApproximatingRX},
$
\mathcal{S}_{i+1}^{N}=\mathcal{H}(\tau,\mathcal{S}_{i}^{N},\epsilon_{h}) \subseteq \exp(\tau A) \mathcal{S}_{i}^{N}\subseteq  \exp(\tau A) S_{i}^{N}=S_{i+1}^{N}.
$
Therefore,  using induction,  $\mathcal{S}_{i}^{N}\subseteq S_{i}^{N}$ for all $i\in \intcc{0;N}$. Similarly, and using Corollary \ref{cor:UnderApproximatingRU}, we have $\mathcal{V}_{i}^{N}\subseteq V_{i}^{N}$ for all $i\in \intcc{0;N}$. Hence, for all $i \in \intcc{0;N}$,  $ \mathcal{W}_{i}^{N}\subseteq W_{i}^{N}$. Moreover, for all  $i \in \intcc{0;N}$,
$\Lambda_{i}^{N}=\mathcal{S}_{i}^{N}+\mathcal{W}_{i}^{N}
 \subseteq {S}_{i}^{N}+{W}_{i}^{N}=\Gamma_{i}^{N}=\mathcal{R}(i\tau)$. Finally, $\bigcup_{i=0}^{N}\Lambda_{i}^{N}\subseteq\bigcup_{i\in \intcc{0;N}}\mathcal{R}(i\tau)\subseteq \bigcup_{t\in \intcc{0,T}}\mathcal{R}(t)= \mathcal{R}(\intcc{0,T})$.
\end{proof}

\begin{remark}[Assumptions on initial and input sets]
\label{rem:Assumptions}
In Section \ref{sec:SystemDescription}, it was assumed that both the initial and input sets are in $\mathbb{A}_{n}$. This assumption can be slightly relaxed if one of these sets is strictly equal to $\{0\}$. For example, if $X_{0}=\{0\}$, then $\mathcal{R}(t)=\mathcal{R}_{u}(t),~t\in \intcc{0,T}$. The proposed method can be implemented in this case by considering computing the sets $\mathcal{W}_{i}^{N},~i\in \intcc{0;N}$, which are independent of $X_{0}$, and omitting the computations of $\mathcal{S}_{i}^{N},~i\in \intcc{0;N}$. Similarly, if $U=\{0\}$, we have $\mathcal{R}(t)=\mathcal{R}_{h}(t),~t\in \intcc{0,T}$, and the proposed method can be implemented by considering the computations of $\mathcal{S}_{i}^{N},~i\in \intcc{0;N}$, only, which are independent of $U$. 
\end{remark}

\section{Implementation using zonotopes and memory complexity}
\label{sec:Implementation}
Fix $N\in \mathbb{N}$ and recall the definitions of $\{\mathcal{S}_{i}^{N}\}_{i=0}^{N}$, $ \{\mathcal{V}_{i}^{N}\}_{i=0}^{N}$, $\{\mathcal{W}_{i}^{N}\}_{i=0}^{N}$ and $\{\Lambda_{i}^{N}\}_{i=0}^{N}$  in Theorem \ref{thm:ProposedMethod}. The computations of $\mathcal{S}_{i}^{N},~\mathcal{V}_{i}^{N},~i\in\intcc{0;N}$ are straightforward for any arbitrary norm on $\mathbb{R}^{n}$  since these sets are simply full-dimensional affine transformations of unit balls. However, the computations of $\{\mathcal{W}_{i}^{N}\}_{i=0}^{N}$ and $\{\Lambda_{i}^{N}\}_{i=0}^{N}$ involve Minkowski sums, whose explicit expressions are generally unknown. If the embedded norm is the maximum norm, then  Minkowski sums can be computed explicitly. Hence, we implement the proposed method using zonotopes, i.e., affine transformations of closed unit balls w.r.t. the maximum norm. 

Given $c\in \mathbb{R}^{n}$, $G\in \mathbb{R}^{n\times p}$, a zonotope $\Zonotope{c}{G}\subseteq \mathbb{R}^{n}$  is defined by  $c+G\mathbf{B}_{p}^{\infty}$, where $\mathbf{B}_{p}^{\infty}$ denotes the $p$-dimensional closed unit ball  w.r.t. the maximum norm. The columns of $G$ are  referred to as the generators of $\Zonotope{c}{G}$ and the ratio $p/n$ is referred to as the order of $\Zonotope{c}{G}$ and is denoted by $o(\Zonotope{c}{G})$ (e.g., the order of $\mathbf{B}_{n}^{\infty}$ is one). Herein,  the number of generators of $\Zonotope{c}{G}$ is denoted by $\Gen{\Zonotope{c}{G}}$.  For any two zonotopes $\Zonotope{c}{G}, \Zonotope{\tilde{c}}{\tilde{G}}\subseteq \mathbb{R}^{n}$ and any linear transformation $L\in \mathbb{R}^{m\times n}$, 
$
\Zonotope{c}{G}+\Zonotope{\tilde{c}}{\tilde{G}}=\Zonotope{c+\tilde{c}}{[G,\tilde{G}]}$ and $
L\Zonotope{c}{G}=\Zonotope{Lc}{LG}.
$ 

Let us analyze the memory complexity of the proposed method implemented with  zonotopes. We have $\Gen{X_{0}}=o(X_{0})n$ and $\Gen{U}=o(U)n$. As affinely transforming zonotopes preserves their orders, we have $\Gen{\mathcal{S}_{i}^{N}}=o(X_{0})n,~i\in \intcc{0;N}$, and $\Gen{\mathcal{V}_{i}^{N}}=o(U)n,~i\in \intcc{0;N}$. The sequence $\{\mathcal{W}_{i}^{N}\}_{i=0}^{N}$ is computed as 
$
\mathcal{W}_{i}^{N}=\sum_{j=0}^{i-1}\mathcal{V}_{j}^{N},~i\in \intcc{0;N}.
$
Hence, 
$
\Gen{\mathcal{W}_{i}^{N}}=i  o(U)n, ~i \in \intcc{0;N}.
$
Consequently, as $\Lambda_{i}^{N}=\mathcal{S}_{i}^{N}+\mathcal{W}_{i}^{N},~i\in \intcc{0;N}$,
$
\Gen{\Lambda_{i}^{N}}=o(X_{0})n+i o(U)n, ~i \in \intcc{0;N}.
$
Finally, the total number of generators from the sequence $\{\Lambda_{i}^{N}\}_{i=0}^{N}$ is $\sum_{i=0}^{N}\Gen{\Lambda_{i}^{N}}=(N+1)o(X_{0})n+N(N+1)o(U)n/2$. This shows that the total number of generators stored is of order $N^2n$ and that gives a  space complexity of order $N^2n^2${ (second order in the argument $N$)}, which is identical, e.g., to the space complexity of the over-approximation method in \cite{SerryReissig21}. {However,  if we  store only the sets $\mathcal{S}_{i}^{N},~\mathcal{V}_{i}^{N}$,~ $i\in \intcc{0;N}$ (the sets prior to Minkowski sum computations), the space complexity is reduced to be of order $Nn^2$ (first order in the argument $N$), where the sets $\mathcal{W}_{i}^{N},~\Lambda_{i}^{N},~i\in \intcc{0;N}$, can be computed afterwards when needed by means of Minkowski sums, which are computationally inexpensive  \cite{AlthoffFrehse16}. This approach  was proposed in \cite{GirardLeGuernicMaler06}  in order to lower memory complexity, also   resulting in a first order memory complexity with respect to the argument $N$}. In Section \ref{sec:RandomlyGenerated}, we  explore empirically the time efficiency  of the proposed method via means of numerical simulations.

As shown above, Minkowski sums of zonotopes increase the number of generators to be stored. This may limit the applicability of our zonotopic implementation  to cases when the discretization parameter $N$ is not significantly large as the proposed method incorporates iterative Minkowski sums. This issue can be overcome by implementing zonotopic order reductions that replace a given zonotope with a zonotopic   under-approximation with  less generators, which lessen the memory requirement, with the price of  reduced accuracy. { For example, a given zonotope can be first inscribed by another zonotope, with a less number of generators, using one of the standard over-approximating order reduction techniques (see, e.g., \cite{KopetzkiSchuermannAlthoff17}), and then   the over-approximating zonotope can be scaled down to be contained in the original zonotope by utilizing linear programming (see \cite{SadraddiniTedrake19}). Furthermore, a zonotope can be under-approximated by  choosing a subset of its generators,  based on some specified criteria, and replacing the chosen generators by their  sum (or difference) \cite{YangOzay21, RaghuramanKoeln22}. In fact, these two approaches have been implemented in the 2021 version of the reachability software CORA \cite{Althoff15}. In Section \ref{sec:NumericalExample}, we will explore the influence of order reduction techniques on the under-approximations obtained using our proposed method.  
}

\section{Convergence}
\label{sec:Convergence}
In Section \ref{sec:ProposedMethod}, we have introduced the proposed method and proved its under-approximate capability.  The proposed method requires to assign values to the parameters $\epsilon_{h}$ and ${\epsilon}_{u}$ in the interval $\intco{0,1}$. The closer the values of $\epsilon_{h}$ and ${\epsilon}_{u}$ to one, the higher the approximation accuracy.  Herein, we show that if we choose $\epsilon_{h}=1-1/N^2$ and $\epsilon_{u}=1-1/N$, then the proposed method  generates convergent under-approximations, as $N$ approaches $\infty$, 
with first-order convergence guarantees, in the sense of the Hausdorff distance, and that is the second main result of this work. %The technical proof, which is similar to those in \cite{SerryReissig21,LeGuernicGirard10,Girard05} is omitted herein due to space limitation and is included in the archived version \cite{SerryLiu22}.
{
\begin{theorem}\label{thm:Convergence}
Let $N\in \mathbb{N}$, $\tau=T/N$, $\epsilon_{h}=1-1/N^{2}$, and $\epsilon_{u}=1-1/N$. Assume $\tau \in \mathbb{I}$ and that $\tau \norm{A}\leq 1$. Define $\{\mathcal{S}_{i}^{N}\}_{i=0}^{N}$, $ \{\mathcal{V}_{i}^{N}\}_{i=0}^{N}$, $\{\mathcal{W}_{i}^{N}\}_{i=0}^{N}$, and $\{\Lambda_{i}^{N}\}_{i=0}^{N}$ as in Theorem \ref{thm:ProposedMethod}. Then, there exist constants $D_{1}, D_{2},D_{3},D_{4}\in \mathbb{R}_{+}$ that are independent of $N$, such that,  for all $i\in \intcc{0;N}$,
$\mathfrak{d}(\mathcal{S}_{i}^{N},\mathcal{R}_{h}(i\tau))\leq D_{1}\tau$, $\mathfrak{d}(\mathcal{W}_{i}^{N},\mathcal{R}_{u}(i\tau))\leq D_{2}\tau$, 
$
\mathfrak{d}(\Lambda_{i}^{N},\mathcal{R}(i\tau))\leq D_{3}\tau,
$
and
$
\mathfrak{d}(\bigcup_{i=0}^{N}\Lambda_{i}^{N},\mathcal{R}(\intcc{0,T}))\leq D_{4}\tau.
$
\end{theorem}
}
Next, we state  some  technical results that are necessary in the proof of  Theorem \ref{thm:Convergence}.  The proofs of Lemmas \ref{lem:ErrorHomogeneous},  \ref{lem:ErrorInput}, and \ref{lem:BoundingSandV} are given in the Appendix.
\begin{lemma}[Semigroup property of reachable sets \cite{Chernousko94}]\label{lem:SemiGroup}
 Given $0\leq a\leq b \leq T$,
 $
 \mathcal{R}(b)=\exp((b-a)A)\mathcal{R}(a)+\int_{0}^{(b-a)}\exp(sA)U\mathrm{d}s.
 $
 \end{lemma}

\begin{lemma}\label{lem:ErrorHomogeneous}
Let $\Omega=c+G\mathbf{B}_{p}\in  \mathbb{A}_{n}$, $\epsilon\in \intco{0,1}$,  and  $t\in \intcc{0,T}$. Assume $t\norm{A}\leq 1$.  Then, 
$$
\mathfrak{d}(\mathcal{H}(t,\Omega,\epsilon),\e^{t A}\Omega)\leq (2(1-\epsilon) +(t\norm{A})^2)\e^{T\norm{A}}\norm{\Omega}.
$$
\end{lemma}

{
\begin{lemma}\label{lem:ErrorInput}
Let $\Omega=c+G\mathbf{B}_{p}\in \mathbb{A}_{n}$, $\epsilon\in \intco{0,1}$, and  $t\in \mathbb{I}\cap \intcc{0,T}$. Assume $t\norm{A}\leq 1$. Then, 
$$
\mathfrak{d}(\mathcal{I}(t,\Omega,{\epsilon}),\int_{0}^{t}\e^{sA}\Omega\mathrm{d}s) \leq 2((1-\epsilon)t+t^2\norm{A})\e^{T\norm{A}}\norm{\Omega}.
$$
\end{lemma}

\begin{lemma}\label{lem:BoundingSandV}
Let $N\in \mathbb{N}$, $\tau=T/N$, and $\{S_{i}^{N}\}_{i=0}^{N}$ and $\{V_{i}^{N}\}_{i=0}^{N}$ be defined as in Equations \eqref{eq:S} and \eqref{eq:V}, respectively. Then,
$
\norm{S_{i}^{N}}\leq \exp(T\norm{A})\norm{X_{0}}$ and
$\norm{V_{i}^{N}}\leq \tau\exp(2T\norm{A})\norm{U}
$
for all $i\in \intcc{0;N}$.
\end{lemma}
}
Now, we are ready to prove Theorem \ref{thm:Convergence}.

\begin{proof}[Proof of Theorem \ref{thm:Convergence}]
Recall the definitions of $\{S_{i}^{N}\}_{i=0}^{N}$, $\{V_{i}^{N}\}_{i=0}^{N}$, $\{W_{i}^{N}\}_{i=0}^{N}$ and $\{\Gamma_{i}^{N}\}_{i=0}^{N}$ in Lemma \ref{lem:TheoreticalRecursiveRelation}. Note that according to Lemma \ref{lem:TheoreticalRecursiveRelation}, $S_{i}^{N}=\mathcal{R}_{h}(i\tau),~W_{i}^{N}=\mathcal{R}_{u}(i\tau)$ and $\Gamma_{i}^{N}=\mathcal{R}(i\tau)$ for all $i\in \intcc{0;N}$. Assume without loss of generality that $A\neq 0$ (the case when $A=0$ is trivial).  %\footnote{ The case when $A=0$ is trivial, where the term $\norm{A}$ in the  estimates derived  in the remainder of the proof can be replaced by  any arbitrary $\delta\in \mathbb{R}_{+}\setminus \{0\}$.}.  
Let $p_{i}=\mathfrak{d}(S_{i}^{N},\mathcal{S}_{i}^{N}),~i\in \intcc{0;N}$. We have $p_{0}=0$ as $S_{0}^{N}=\mathcal{S}_{0}^{N}=X_{0}$. For $i\in \intcc{1;N}$, we have, using the definitions of $S_{i}^{N}$ and $\mathcal{S}_{i}^{N}$ in Equations \eqref{eq:S} and  \eqref{eq:MathcalS}, respectively, the  triangle inequality, and Lemma \ref{lem:HausdorffDistance}(b),
\begin{align*}
p_{i}&\leq \mathfrak{d}(\e^{\tau A}S_{i-1}^{N},\e^{\tau A}\mathcal{S}_{i-1}^{N})+\mathfrak{d}(\e^{\tau A}\mathcal{S}_{i-1}^{N},\mathcal{H}(\tau,\mathcal{S}_{i-1}^{N},\epsilon_{h}))\\
&\leq \e^{\tau \norm{A}}p_{i-1}+\mathfrak{d}(\e^{\tau A}\mathcal{S}_{i-1}^{N},\mathcal{H}(\tau,\mathcal{S}_{i-1}^{N},\epsilon_{h})).
\end{align*}
Using Lemma \ref{lem:ErrorHomogeneous} and the fact that $\epsilon_{h}=1-(\tau/T)^2$, the term $\mathfrak{d}(\exp(\tau A)\mathcal{S}_{i-1}^{N},\mathcal{H}(\tau,\mathcal{S}_{i-1}^{N},\epsilon_{h}))$ is bounded  above by
$C_{1} \norm{\mathcal{S}_{i-1}^{N}} \tau^{2},
$
where
$
C_{1}=({2}/{T^{2}}+\norm{A}^{2})\exp(T \norm{A}).
$
Moreover, using the fact that $\mathcal{S}_{i-1}^{N}\subseteq S_{i-1}^{N},~i\in \intcc{1;N}$, as shown in  Theorem \ref{thm:ProposedMethod}, and  Lemma \ref{lem:BoundingSandV}, we have 
$
\norm{\mathcal{S}_{i-1}^{N}}\leq \norm{{S}_{i-1}^{N}}\leq M=\exp(T\norm{A})\norm{X_{0}},~i\in \intcc{1;N+1}.
$
Therefore,
$
p_{i}\leq \exp(\tau \norm{A})p_{i-1}+C_{1}M\tau^{2},~i\in \intcc{1;N}.
$
Using induction, we have, for all $i\in \intcc{0;N}$,
\begin{equation}\label{eq:p}
p_{i}\leq \frac{\e^{i\tau \norm{A}}-1}{\e^{\tau \norm{A}}-1}C_{1}M\tau^{2}\leq \frac{\e^{T \norm{A}}-1}{\tau \norm{A}}C_{1} M\tau^{2}=D_{1}\tau,
\end{equation}
where
$
D_{1}= (\exp(T \norm{A})-1)C_{1}M/\norm{A}.
$
Similarly, let $q_{i}=\mathfrak{d}(V_{i}^{N},\mathcal{V}_{i}^{N}),~i\in \intcc{0;N}$. We have, using Lemma \ref{lem:ErrorInput} and the fact that ${\epsilon}_{u}=1-\tau/T$, 
$
q_{0}=\mathfrak{d}(\int_{0}^{\tau}\exp(sA)U\mathrm{d}s, \mathcal{I}(\tau,U,{\epsilon}_{u}))\leq C_{2}\tau^2, 
$
where
$
C_{2}=2(1/T+\norm{A})\exp(T\norm{A})\norm{U}.
$
The remaining terms,  $q_{i},~i\in \intcc{1;N}$, can be bounded using the triangular inequality and Lemma \ref{lem:HausdorffDistance}(b), where we deduce the recursive inequality 
$
q_{i}\leq \e^{\tau \norm{A}}q_{i-1}+\mathfrak{d}(\exp(\tau A)\mathcal{V}_{i-1}^{N},\mathcal{H}(\tau,\mathcal{V}_{i-1}^{N},\epsilon_{h})),~i\in \intcc{1;N}.
$
Now, for the term $\mathfrak{d}(\exp(\tau A)\mathcal{V}_{i-1}^{N},\mathcal{H}(\tau,\mathcal{V}_{i-1}^{N},\epsilon_{h}))$, we use Lemma \ref{lem:ErrorHomogeneous}, which results in the inequality  
$
\mathfrak{d}(\exp(\tau A)\mathcal{V}_{i-1}^{N},\mathcal{H}(\tau,\mathcal{V}_{i-1}^{N},\epsilon_{h})) \leq C_{1} \norm{\mathcal{V}_{i-1}^{N}} \tau^{2}.
$
 Using $\mathcal{V}_{i-1}^{N}\subseteq V_{i-1}^{N},~i\in \intcc{1;N}$ from Theorem \ref{thm:ProposedMethod}, and  Lemma \ref{lem:BoundingSandV},  the sequence $\{\norm{\mathcal{V}_{i-1}^{N}}\}_{i=1}^{N+1}$ is bounded above by $\tau \tilde{M}$, where 
$
\tilde{M}= \exp(2T\norm{A})\norm{U}.
$
Hence,
$
q_{i}\leq \exp(\tau \norm{A})q_{i-1}+C_{1} \tilde{M} \tau^{3},~i\in\intcc{1;N},
$
and by induction, we obtain, for all $i\in \intcc{0;N}$,
\begin{equation}\label{eq:q}
\begin{split}
q_{i}&\leq \e^{i\tau \norm{A}}C_{2}\tau^2+ \frac{\e^{i\tau \norm{A}}-1}{\e^{\tau \norm{A}}-1}C_{1}\tilde{M}\tau^{3}\\
& \leq \e^{i\tau \norm{A}}C_{2}\tau^2+ \frac{\e^{i\tau \norm{A}}}{\tau \norm{A}}C_{1}\tilde{M}\tau^{3}=\e^{i\tau \norm{A}}C_{3}\tau^{2},
\end{split}
\end{equation}
where
$
C_{3}=C_{2}+{C_{1}\tilde{M}}/{\norm{A}}.
$
Next, define $r_{i}=\mathfrak{d}(W_{i}^{N},\mathcal{W}_{i}^{N}),~i \in \intcc{0;N}$. Then, $r_{0}=0$ as $W_{0}^{N}=\mathcal{W}_{0}^{N}=\{0\}$ and, for $i\in \intcc{1;N}$, 
$
r_{i}\leq r_{i-1}+q_{i-1},
$
where we have utilized Lemma \ref{lem:HausdorffDistance}(a). 
Using induction,  the sequence $\{r_{i}\}_{i=0}^{N}$  is bounded above as follows:
$
r_{i}\leq \sum_{j=0}^{i-1}q_{j},~i\in \intcc{0;N},
$
where $\sum_{j=0}^{-1}(\cdot)=0$. 
Hence, using  estimate \eqref{eq:q},
\begin{equation}\label{eq:r}
\begin{split}
 r_{i}&\leq \sum_{j=0}^{i-1}\e^{j\tau \norm{A}}C_{3}\tau^2
 \leq \frac{\e^{i\tau \norm{A}}-1}{\e^{\tau \norm{A}}-1} C_{3}\tau^{2}\\
 &\leq  \frac{\e^{T \norm{A}}-1}{\tau \norm{A}} C_{3}\tau^{2}= D_{2}\tau,
 \end{split}
 \end{equation}
  where 
  $
 D_{2}= (\exp(T \norm{A})-1)C_{3}/\norm{A}.
  $
Let $s_{i}=\mathfrak{d}(\Gamma_{i}^{N},\Lambda_{i}^{N}),~i\in \intcc{0;N}$. Using Lemma \ref{lem:HausdorffDistance}(a), we have, for all $i\in \intcc{0;N}$, 
$
s_{i}\leq p_{i}+r_{i}.
$
By incorporating the bounds \eqref{eq:p}  and \eqref{eq:r}, we get, for $i\in \intcc{0;N}$, 
\begin{equation}\label{eq:s}
s_{i}\leq  D_{1}\tau+D_{2}\tau  = D_{3} \tau,
\end{equation}
where $D_{3}={D_{1}} +D_{2}$.

Now, we prove the last claim of the theorem. Note that, using the definition of reachable sets given in Equation \eqref{eq:R(t)}, 
\begin{equation}\label{eq:K}
\norm{\mathcal{R}(t)}\leq K\defas \e^{T\norm{A}}\left(\norm{X_{0}}+T\norm{U} \right),~t\in \intcc{0,T}.
\end{equation}
Define $\tilde{\mathcal{R}}_{N}(t)=\Lambda_{\tilde{i}(t)}^{N}$,  where $\tilde{i}(t)=\lfloor t/\tau \rfloor$ (floor of $t/\tau $). Note that $\bigcup_{i=0}^{N} \Lambda_{i}^{N}=\bigcup_{t\in \intcc{0,T}}\tilde{\mathcal{R}}_{N}(t)$ and that, using Theorem \ref{thm:ProposedMethod}, $\tilde{\mathcal{R}}_{N}(t)\subseteq \mathcal{R}(t),~t \in \intcc{0,T}.$  Using Lemma \ref{lem:HausdorffDistance}(d), the Hausdorff distance between $\bigcup_{i=0}^{N} \Lambda_{i}^{N}$ and $\mathcal{R}(\intcc{0,T})$ satisfies the inequality 
\begin{equation}\label{eq:HDReachTube}
\mathfrak{d}(\bigcup_{i=0}^{N} \Lambda_{i}^{N}, \mathcal{R}(\intcc{0,T}))\leq \sup_{t\in \intcc{0,T}} \mathfrak{d}(\tilde{\mathcal{R}}_{N}(t),\mathcal{R}(t)).
\end{equation}
Let  $t\in \intcc{0,T}$ and set $i=\lfloor t / \tau\rfloor $. Then,  $\tilde{\mathcal{R}}_{N}(t)=\Lambda_{i}^{N}$. Using the triangular inequality,
$$
\mathfrak{d}(\mathcal{R}(t),\tilde{\mathcal{R}}_{N}(t))
\leq \mathfrak{d}(\mathcal{R}(t),\mathcal{R}(i\tau))+\mathfrak{d}(\mathcal{R}(i\tau),\Lambda_{i}^{N}).
$$ 
Let us estimate $\mathfrak{d}(\mathcal{R}(t),\mathcal{R}(i\tau))$. Note that $0\leq t-i\tau \leq \tau\leq T$.  Using  Lemma \ref{lem:SemiGroup}, $\mathcal{R}(t)$ can be written as
$
\mathcal{R}(t)= \exp((t-i\tau)A)\mathcal{R}(i\tau)+ \int_{0}^{(t-i\tau)}\exp(sA)U\mathrm{d}s.
$
Hence, using Lemma \ref{lem:HausdorffDistance}(a),(c) and estimates \eqref{eq:BoundonExpAndL} and \eqref{eq:K} (below, $\Delta$ denotes $t-i\tau$),
\begin{align*}
\mathfrak{d}(\mathcal{R}(t),\mathcal{R}(i\tau))
&\leq \mathfrak{d}(\e^{\Delta A}\mathcal{R}(i\tau),\mathcal{R}(i\tau))+\mathfrak{d}(\int_{0}^{\Delta}\e^{sA}U\mathrm{d}s,0)\\
&\leq
\norm{\e^{\Delta A}-\id}\norm{\mathcal{R}(i\tau)}
+
\norm{U}\int_{0}^{\Delta}\e^{s \norm{A}}\mathrm{d}s\\
&\leq 
K\Delta\norm{A}\e^{\Delta\norm{A}}+\norm{U} \Delta\e^{T \norm{A}}\leq C_{4}\tau,
\end{align*}
where $C_{4}=(K\norm{A}+\norm{U})\exp(T \norm{A})$. Moreover, using \eqref{eq:p}, we have, $\mathfrak{d}(\mathcal{R}(i\tau),\Lambda_{i}^{N})\leq C_{4}\tau$.  Therefore,  
$
\mathfrak{d}(\mathcal{R}(t),\tilde{\mathcal{R}}_{N}(t))\leq D_{4} \tau,
$
where $D_{4}=C_{4}+D_{3}$. As the choice of  $t\in \intcc{0,T}$ is arbitrary and in view of \eqref{eq:HDReachTube}, the proof is complete.
\end{proof}

\section{Numerical Examples}
\label{sec:NumericalExample}
In this section, we illustrate  the proposed method through three numerical examples. The proposed method is implemented using zonotopes in MATLAB (2019a) and run on an AMD Ryzen 5 2500U/2GHz processor. { Plots of zonotopes, scaling-based under-approximations \cite{KochdumperAlthoff20}, and reduced-order zonotopic under-approximations are produced using the software CORA (2021 version) \cite{Althoff15}}. For all the considered linear systems  and  values of the time discretization parameter $N$,  $\int_{0}^{T/N}\exp(sA)\mathrm{d}s$ is invertible. The optimization problems associated with evaluating the parameters $k_{\min}$, $\kappa$, and $\eta$, given in Equations \eqref{eq:Kmin}, \eqref{eq:Kappa}, and \eqref{eq:Eta}, respectively, are solved via brute force. The invertibility condition used in the definition of $\eta$ in \eqref{eq:Eta} is checked using the $\mathtt{rank}$ function in MATLAB.
\subsection{2-D system with a closed-form reachable set}
\label{sec:2D}
Consider an instance of system \eqref{eq:LinearSystem} (perturbed double integrator system), with
$
A=\big(\begin{smallmatrix}
  0 & 0\\
  1 & 0
\end{smallmatrix}\big),~
 U=\intcc{0,1}\times \intcc{0,1}, ~X_{0}=\{(0,0)^{\intercal}\},
$
and 
 $T=1$.
Note that $\mathcal{R}(T)=\mathcal{R}_{u}(T)=\int_{0}^{T}\exp(sA) U_{1}\mathrm{d}s+\int_{0}^{T}U_{2}\mathrm{d}s, 
$
where $U_{1}=\intcc{0,1}\times \{0\}$, and $U_{2}=\{0\}\times \intcc{0,1}$. The set $\int_{0}^{T}\exp(sA) U_{1}\mathrm{d}s$ is given explicitly as  $\int_{0}^{T}\exp(sA) U_{1}\mathrm{d}s=\{(x,y)^{\intercal}\in \mathbb{R}^{2},~x^2/2\leq y \leq x-x^2/2,~x\in \intcc{0,1}\}$ (see \cite[the formula of $M_2$, p.~363]{Ferretti97}), whereas, using \cite[Theorem~3,~p.~21]{AubinCellina84}, $\int_{0}^{T}U_{2}\mathrm{d}s=\{0\}\times\intcc{0,1}$. Hence, 
$
\mathcal{R}(T)=\Set{(x,y)^{\intercal}\in \mathbb{R}^{2}}{x^2/2\leq y \leq x-x^2/2+1,~x\in \intcc{0,1}}$. In view of Remark \ref{rem:Assumptions}, we aim to compute under-approximations of $\mathcal{R}(T)$ using the proposed method. Herein, we consider different values of the discretization parameters $N$  and set $\epsilon_{h}=1-1/N^{2}$ and $\epsilon_{u}=1-1/N$.  
\begin{figure}
    \centering
    \includegraphics[width=0.9\columnwidth]{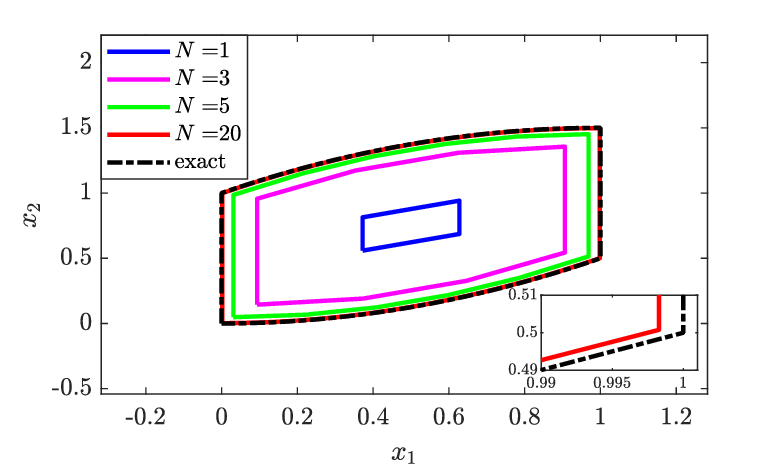}
    \caption{Under approximations of the reachable set of the double integrator system in Section \ref{sec:2D} with different values of the discretization parameter $N$. }
    \label{fig:2D}
\end{figure}

Figure \ref{fig:2D} displays several under-approximations of $\mathcal{R}(T)$ with $N\in \{1,3,5,20\}$. The mentioned figure shows how the obtained approximations are indeed enclosed by the exact reachable set, in agreement with Theorem \ref{thm:ProposedMethod}. Moreover, the mentioned Figure exhibits how the approximation accuracy of the proposed method  increases as $N$ increases, which further supports the convergence result in Theorem \ref{thm:Convergence}. { The computational time associated with evaluating the under-approximation with $N=20$ is less than 0.003 seconds.}
\subsection{5-D system}
\label{sec:5D}
{ Herein, we adopt a five dimensional instance of system \eqref{eq:LinearSystem} from the literature, where matrix $A$ and the sets $X_{0}$ and $U$ are given  in \cite[Equation~3.11,~p.~39]{Althoff10}, and set  $T=1$. We aim to under-approximate the reachable tube $\mathcal{R}(\intcc{0,T})$ using the proposed method with $N\in \{10,100\}$, where we set $\epsilon_{h}=1-1/N^{2}$ and $\epsilon_{u}=1-1/N$, and compute the sets $\Lambda_{i}^{N},~i\in \intcc{1;N}$ ($\cup_{i=0}^{N}\Lambda_{i}^{N}$ is the desired under-approximation  herein)}. As the exact reachable tube is not known, we use the  convergent over-approximation method of  Serry and Reissig \cite{SerryReissig21}, with a refined time discretization (200 steps) and an accurate approximation of the matrix exponential, $\mathcal{L}(\cdot,10)$, to produce an accurate representation of the reachable tube and use it as the basis of comparison. { Furthermore, we additionally compute, for the case $N=10$, a modified version of the  under-approximation from the proposed method, where each set $\Lambda_{i}^{N}$ is under-approximated using order reduction based on summing generators  (method \texttt{sum} in CORA). The zonotopes  resulting from order reduction   are at most of order 2.} 

\begin{figure}
 \centering
 \includegraphics[width=\columnwidth]{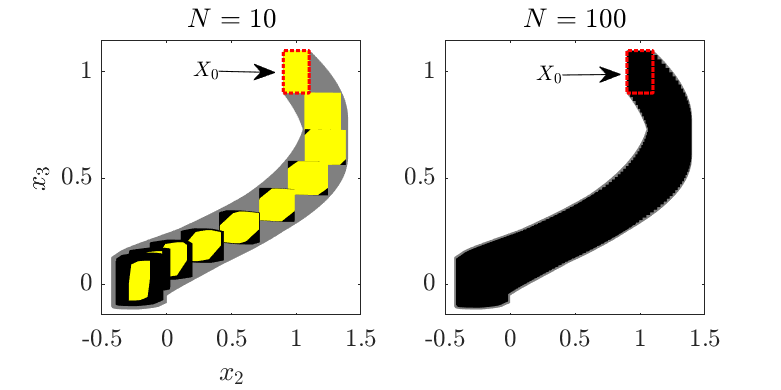}
  \includegraphics[width=\columnwidth]{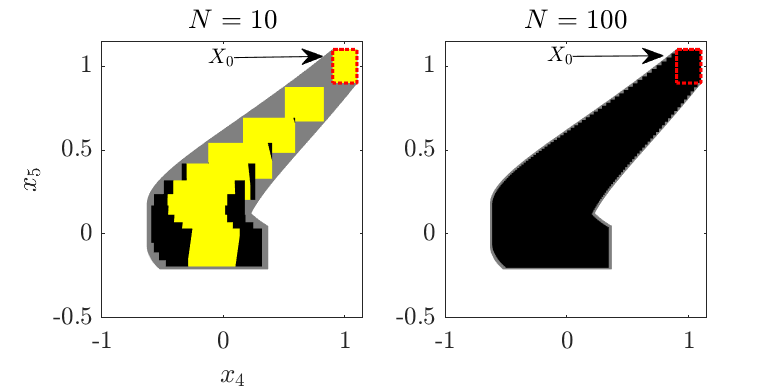}
    \caption{{ $x_2-x_3$ (top) and $x_4-x_5$ (bottom) projections of the over-approximation  (grey) of the reachable tube $\mathcal{R}(\intcc{0,T})$ of the 5-D system  in Section \ref{sec:5D} using the method in \cite{SerryReissig21}, the under-approximations from the proposed method given by the sets $\Lambda_{i}^{N},~i\in \intcc{0;N}$ (black),  with $N=10$ (left) and $N=100$ (right),  and reduced-order under-approximations (yellow) of the sets $\Lambda_{i}^{N},~i\in \intcc{0;N}$ in the case $N=10$ .}} 
    \label{fig:5D}
\end{figure}
Figure \ref{fig:5D} displays two projections of the over-approximation and the under-approximations. As seen in the mentioned figure, the over-approximation (grey area) encloses the under-approximations (black areas) from the proposed method.  { Moreover, the reduced-order zonotopes (yellow areas) are enclosed by  the sets $\Lambda_{i}^{N},~i\in \intcc{0;N}$ from the proposed method (without the order reduction). We observe that   for the case $N=10$, the sets $\Lambda_{i}^{N},~i\in \intcc{0;N}$  and their corresponding reduced-order  under-approximations are almost over-lapping for the first few iterations  ($i$ is small); however, the accuracy of the reduced-order under-approximations decays with each iteration due to their constrained order (2 in this case). This highlights the trade-off between accuracy and memory reduction when it comes to incorporating order reduction techniques in computing under-approximations.}    We easily observe that the under-approximation when $N=100$ resembles the over-approximation more accurately, relative to the case when $N=10$. 
This implies that  the under-approximation with $N=100$ also resembles the actual reachable tube more accurately, and this is a consequence of the convergence guarantees  of  Theorem \ref{thm:Convergence}.  { The computational times associated with evaluating the under-approximations from the proposed method (without order reduction) with $N=10$ and $N=100$ are  0.0199 and 0.0279 seconds, respectively. }  
\subsection{Randomly generated systems}
\label{sec:RandomlyGenerated}
{ In this section, we study the performance of the proposed method on randomly generated linear systems, where the matrix exponentials associated with the generated systems are not known exactly}.
{
\subsubsection{Homogeneous linear systems}
\label{sec:Comparison}
To compare the performance of the proposed method with the scaling method \cite{KochdumperAlthoff20}, we consider homogeneous linear systems ($\dot{x}=Ax$), with randomly generated instances of matrix $A$,  using the MATLAB command \texttt{rand}, where $n\in \{2,4,6,8,10\}$,  $X_{0}=\mathbf{B}_{n}^{\infty}$, and $T=1$. For each instance of $A$, we compute under-approximations of $\mathcal{R}(T)=\mathcal{R}_{h}(T)$ using both the proposed method (see Remark \ref{rem:Assumptions}) and the scaling method. The scaling method is implemented in the 2021 version of CORA, with settings tuned and approved by the first author of \cite{KochdumperAlthoff20}. 
For our proposed method, we use $N=100$, and set $\epsilon_{h}=1-1/N^{2}$. The  computational times and the volumes of the under-approximations from the two methods  are listed in Table \ref{tab:ComparisonCora}. 
\begin{table}
    \scriptsize
    \centering
        \caption{Computational times and  volumes of  under-approximations  from the proposed  (with subscript ${p}$) and scaling  \cite{KochdumperAlthoff20} (with subscript ${s}$) methods  for randomly generated homogeneous linear systems.}
    \label{tab:ComparisonCora}
    \begin{tabular}{|c|c|c|c|c|c|}
     \hline 
        $n$ &2  & 4 & 6  & 8 & 10\\
             \hline 
 $t_{\mathrm{c},p}$ [s]& 
    0.0150 & 0.0156 & 0.0138 & 0.0120 & 0.0179\\
         \hline 

    $t_{\mathrm{c},s}$ [s]& 
   1.5294 & 5.3289 & 16.7495 & 43.5934 & 96.3664\\
         \hline 
$\mathrm{vol}_{p}$&
   13.7889 & 94.3212 & 748.5102 & 5.5201e+03 & 1.0776e+05\\
    \hline
$\mathrm{vol}_{s}$&  
    13.5197 & 91.1578 & 715.6002 & 5.1814e+03 & 9.9622e+04 \\
    \hline
    \end{tabular}
\end{table} 
% $$
%     '\left(\begin{array}{ccccc} 0.0150 & 0.0156 & 0.0138 & 0.0120 & 0.0179\\ 1.5294 & 5.3289 & 16.7495 & 43.5934 & 96.3664\\ 13.7889 & 94.3212 & 748.5102 & 5.5201e+03 & 1.0776e+05\\ 13.5197 & 91.1578 & 715.6002 & 5.1814e+03 & 9.9622e+04 \end{array}\right)'
% $$
The table  exhibits that the proposed method performs marginally better than the scaling method in terms of accuracy, with slightly larger volumes for the computed under-approximations. Most importantly, Table \ref{tab:ComparisonCora} displays how the proposed method outperforms the scaling method in terms of computational time while having better/comparable accuracy. This may be due to the fact that the proposed method obtains under-approximations by scaling intermediate sets, using simple optimization problems (Equations \eqref{eq:Kappa} and \eqref{eq:Eta}), without the need to solve optimization problems that utilize enclosures of  boundaries of reachable sets as in the scaling method.

\subsubsection{Linear systems with input}
Herein, we study  empirically the time efficiency  associated with obtaining
the under-approximations $\Lambda_{i}^{N},~ i\in \intcc{1;N},$ from the proposed
method. We
 consider instances of system \eqref{eq:LinearSystem}, with $n$ ranging between
10 and 200, where  $X_{0} =U= \mathbf{B}_{n}^{\infty}$
, and
$T = 1$. For each instance of $n$, we randomly generate  a corresponding matrix $A$ using the MATLAB command \texttt{rand}. The
eigenvalues of each generated $A$ are checked not to be purely imaginary
to ensure the invertibility of
$\int_{0}^{T/N}
 \exp(sA)\mathrm{d}s$ (see Lemma
\ref{lem:IntegralInvertible}). For every $n$, we estimate the  computational time associated
with implementing the proposed method, 
where we set $N=100$ and $\epsilon_{h}=\epsilon_{u}=0.8$.  

Through our empirical exploration of the method performance, we noticed that for the randomly generated systems and set time interval,   the  reachable sets can be ``narrow'' especially when the dimension is high. Subsequently, the  under-approximations obtained from the proposed method are nearly degenerate (ill-conditioned). This  negatively influences the computations of the deflation parameter as the term $\norm{G}\norm{G^{\dagger}}$ (condition number) in Equation \eqref{eq:lambda} becomes  significantly large.   Consequently, the order of the approximation $\mathcal{L}$, which is determined using the definition of $\kappa$ in Equation \eqref{eq:Kappa},  may have to be substantially large in order for the deflation parameter values to be larger than the specified design parameter $\varepsilon_{h}$. To avoid these degenerate cases, we restrict our investigation herein to normalized system matrices, where each  generated  matrix $A$ is divided by its maximum norm.   We note that  normalized system matrices were considered in previous studies that investigated the performance of  over-approximation methods (see, e.g., \cite{Girard05}).  
\begin{figure}
     \centering
     \includegraphics[width=0.9\columnwidth]{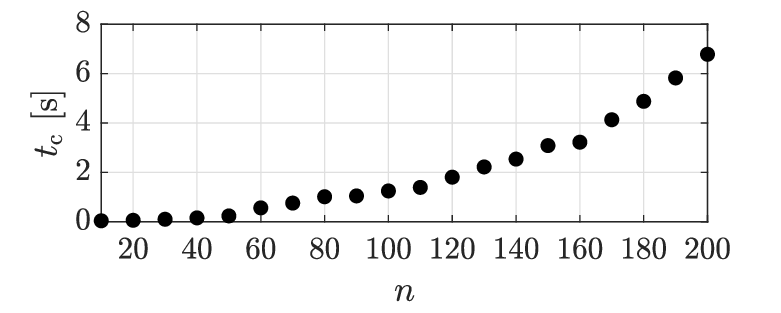}
     \caption{{ Computational time associated with evaluating $\{\Lambda_{i}\}_{i=1}^{N}$  as a  function of $n$ for (normalized) randomly generated instances of system \eqref{eq:LinearSystem}.}}
    \label{fig:tcpu}
 \end{figure}

Figure \ref{fig:tcpu} plots the recorded  computational time, associated with the computations from the proposed method, as a function of the dimension $n$.
{ The mentioned figure exhibits the efficiency of the computations, for the considered randomly generated systems, as the computational time is less than 0.25 seconds for $n\leq 50$, less than 1.25 seconds   for $50 \leq n\leq 100$, and  less than 7 seconds when $n = 200$. This numerical example highlights the potential role of the proposed method in real-time computations in different applications, such as falsification and control synthesis as highlighted in Remark \ref{rem:Applications}, due to the relatively fast computations especially for moderate dimensions ($n\leq50$).}

}

\section{Conclusion}
\label{sec:Conclusion}
In this paper, we proposed a novel convergent method to under-approximate finite-time forward reachable sets and tubes of a class of continuous-time linear uncertain systems, where approximations of the matrix exponential and its integral are utilized. In future work, we aim to explore  extensions and modifications of the proposed  method to cover wider classes of systems, reduce computational cost, and increase accuracy. Furthermore, we seek to address how to obtain under-approximations   in cases when the reachable sets are ``almost'' degenerate.

\section*{Aknowledgement}
The authors thank Niklas Kochdumper (Stony Brook University, USA) for his useful guidance and for tuning the settings used for the scaling method implementation in CORA.

\begin{appendix}
We will need the following lemma in the proofs of Lemmas \ref{lem:ErrorHomogeneous} and \ref{lem:ErrorInput}.
\begin{lemma}\label{lem:BoundingGwithOmega}
Given $\Omega=c+G\mathbf{B}_{p}\subseteq \mathbb{R}^{n}$, where $c\in \mathbb{R}^{n}$ and $G\in \mathbb{R}^{n\times p}$, we have
$\norm{c}\leq \norm{\Omega}$ and $\norm{G}=\norm{G\mathbf{B}_{p}}\leq 2\norm{\Omega}$.
\end{lemma}
\begin{proof}
The first inequality follows from the fact that $c\in \Omega$. The second inequality is deduced as follows:
$
\norm{G\mathbf{B}_{p}}=\sup_{b\in \mathbf{B}_{p}}\norm{Gb}\leq \sup_{b\in \mathbf{B}_{p}}\norm{Gb+c}+\norm{c}=\norm{\Omega}+\norm{c}\leq 2\norm{\Omega}.
$
\end{proof}

\begin{proof}[Proof of Lemma \ref{lem:ErrorHomogeneous}]
For convenience, define $H_{1}=\mathcal{H}(t,\Omega,\epsilon), H_{2}=\mathcal{L}(t,\kappa(t,\Omega,\epsilon))\Omega$, and $H_{3}=\exp(t A)\Omega$. 
Using the definition of $\mathcal{H}$ in Equation \eqref{eq:H}, the triangle inequality,
\begin{equation}\label{eq:dH1H3TriangleInequality}
\mathfrak{d}(H_{1},H_{3}) \leq \mathfrak{d}(H_{1},H_{2}) +\mathfrak{d}(H_{2},H_{3}).
\end{equation}
Using Lemma \ref{lem:HausdorffDistance}(a),(c), we have
\begin{align*}
\mathfrak{d}(H_{1},H_{2}) \leq
 \abs{1-\lambda(t,\Omega, \kappa(t,\Omega,\epsilon_{h}))}\norm{\mathcal{L}(t,\kappa(t,\Omega,\epsilon))}\norm{G\mathbf{B}_{p}}.
\end{align*}
Note that, by the definition of $\kappa$ given in Equation \eqref{eq:Kappa}, $\epsilon<\lambda(t,\Omega, \kappa(t,\Omega,\epsilon))< 1$, which  indicates that $\abs{1-\lambda(t,\Omega, \kappa(t,\Omega,\epsilon))} \leq 1-\epsilon.
$
Furthermore, using \eqref{eq:BoundonExpAndL}, $\norm{\mathcal{L}(t,\kappa(t,\Omega,\epsilon_{h}))}\leq \exp(t\norm{A})$.  Moreover, we have, using Lemma \ref{lem:BoundingGwithOmega}, $\norm{G\mathbf{B}_{p}}\leq 2 \norm{\Omega}$. Hence,
\begin{equation}\label{eq:dH1H2}
\mathfrak{d}(H_{1},H_{2})\leq 2(1-\epsilon)\e^{t\norm{A}}\norm{\Omega}.
\end{equation}
Next, we estimate $\mathfrak{d}(H_{2}, H_{3})$, which, using Lemma \ref{lem:HausdorffDistance}(c), satisfies the inequality
$$
\mathfrak{d}(H_{2},H_{3})\leq \norm{\e^{t A}-\mathcal{L}(t,\kappa(t,\Omega,\epsilon))}\norm{\Omega}.
$$
As $\mathcal{L}(t,\kappa(t,\Omega,\epsilon))$ is a Taylor approximation of $\exp(t A)$ of, at least,  first order  ($\kappa(t,\Omega,\epsilon)\geq 2$), and that $t\norm{A}\leq 1$, we have, using the bound \eqref{eq:BoundonApproxError}, where the arguments of $\kappa$ are dropped for convenience,
\begin{align*}
\norm{\e^{t A}-\mathcal{L}(t,\kappa)}& \leq {(t\norm{A})^\kappa}\frac{\e^{t\norm{A}}}{\kappa!}\\
&\leq {(t\norm{A})^2}\frac{\e^{t\norm{A}}}{2!}\\
&\leq (t\norm{A})^2\e^{t\norm{A}}.
\end{align*}
 Consequently, 
\begin{equation}\label{eq:dH2H3}
 \mathfrak{d}(H_{2},H_{3})\leq (t\norm{A})^2\e^{t\norm{A}}\norm{\Omega}. 
\end{equation}
Combining estimates \eqref{eq:dH1H3TriangleInequality}, \eqref{eq:dH1H2}, and \eqref{eq:dH2H3} yields 
\begin{align*}
\mathfrak{d}(\mathcal{H}(t,\Omega,\epsilon),\e^{t A}\Omega) %\leq & 2(1-\epsilon)\e^{t\norm{A}}\norm{\Omega}+(t\norm{A})^2\e^{t\norm{A}}\norm{\Omega}\\
\leq& (2(1-\epsilon) +(t\norm{A})^2)\e^{t\norm{A}}\norm{\Omega}\\
\leq & (2(1-\epsilon) +(t\norm{A})^2)\e^{T\norm{A}}\norm{\Omega}.
\end{align*}
\end{proof}

\begin{proof}[Proof of Lemma \ref{lem:ErrorInput}]
For convenience, define $I_{1}=\mathcal{I}(t,\Omega,{\epsilon}),~I_{2}=\int_{0}^{t}\exp(sA)\mathrm{d}s\Omega$, and $I_{3}=\int_{0}^{t}\exp(sA)\Omega\mathrm{d}s$. Using the triangular inequality, 
\begin{equation}\label{eq:dI1I3TriangleInequality}
\mathfrak{d}(I_{1},I_{3}) \leq \mathfrak{d}\left(I_{1},I_{2}\right)+\mathfrak{d}\left(I_{2}, I_{3}\right).
\end{equation}
The term
$\mathfrak{d}\left(I_{1}, I_{2}\right)$ can be bounded above using the triangular inequality, the definition of $\mathcal{I}$ given in Equation \eqref{eq:I}, and Lemma \ref{lem:HausdorffDistance}(a),(c), as follows:
\begin{align*}
\mathfrak{d}(I_{1}, I_{2}) \leq& \mathfrak{d}(I_{1},\mathcal{T}(t,\eta(t,\Omega,{\epsilon}))\Omega)+\mathfrak{d}(\mathcal{T}(t,\eta(t,\Omega,{\epsilon}))\Omega, I_{2})\\
 \leq&  \abs{1-\lambda(t,\Omega, \eta(t,\Omega,{\epsilon}))}\norm{\mathcal{T}(t,\eta(t,\Omega,{\epsilon}))}\norm{G\mathbf{B}_{p}}\\
 &+ \norm{\int_{0}^{t}\e^{sA}\mathrm{d}s-\mathcal{T}(t,\eta(t,\Omega,{\epsilon}))}\norm{\Omega}.
\end{align*}
Using estimate \eqref{eq:BoundonExpAndL}, we obtain the bound
\begin{align*}
\norm{\mathcal{T}(t,\eta(t,\Omega,{\epsilon}_{u}))}&\leq \int_{0}^{t}\norm{\mathcal{L}(s,\eta(t,\Omega,{\epsilon}_{u}))}\mathrm{d}s\\
&\leq \int_{0}^{t}\e^{s\norm{A}}\mathrm{d}s\leq t \e^{t\norm{A}}.
\end{align*}
Moreover, using the definition of $\eta$ in Equation \eqref{eq:Eta}, we have
$
\abs{1-\lambda(t,\Omega, \eta(t,\Omega,{\epsilon}))}\leq 1-\epsilon.
$
Also, using Lemma \ref{lem:BoundingGwithOmega}, we have $\norm{G\mathbf{B}_{p}}\leq 2\norm{\Omega}$. Besides that, using estimate \eqref{eq:BoundonApproxError} and the definition of $\eta$, we have (the arguments of $\eta$ are dropped for convenience)
\begin{align*}
\norm{\int_{0}^{t}\e^{sA}\mathrm{d}s-\mathcal{T}(t,\eta)}& \leq \int_{0}^{t}\norm{\e^{sA}-\mathcal{L}(s,\eta)}\mathrm{d}s\\
&\leq\int_{0}^{t}\frac{(s\norm{A})^{\eta}}{\eta!}\e^{s\norm{A}}\mathrm{d}s\\
&\leq  \frac{t\norm{A}}{\eta!}\e^{t\norm{A}} \int_{0}^{t}\mathrm{d}s\\
&\leq t^{2}\norm{A}\e^{t\norm{A}},
\end{align*}
where we have used the facts that $\eta\geq 1$ and $s\norm{A}\leq t\norm{A}\leq 1,~s\in \intcc{0,t}$. Therefore,
\begin{equation}\label{eq:dI1I2}
\begin{split}
\mathfrak{d}(I_{1}, I_{2})&\leq  2(1-\epsilon)t{\e^{t\norm{A}}}\norm{\Omega}+ t^2\norm{A}\e^{t\norm{A}}\norm{\Omega}\\
 &= (2(1-\epsilon)t+t^2\norm{A})\e^{t\norm{A}}\norm{\Omega}.  
\end{split}
\end{equation}
Now, we estimate $\mathfrak{d}\left(I_{2}, I_{3}\right)$. 
Using \cite[Theorem~3,~p.~21]{AubinCellina84}, $I_{2}$ can be  rewritten as
\begin{align*}
I_{2}&=(\frac{1}{t}\int_{0}^{t}\e^{sA}\mathrm{d}s) (t\Omega)=
(\frac{1}{t}\int_{0}^{t}\e^{sA}\mathrm{d}s) \int_{0}^{t}\Omega \mathrm{d}s
\\
&=\int_{0}^{t}B\Omega \mathrm{d}s,
\end{align*}
where 
$
B=({1}/{t})\int_{0}^{t}\exp(sA)\mathrm{d}s.
$
Then, the Hausdorff distance between $I_{2}$ and $I_{3}$ can be estimated, using Lemma \ref{lem:HausdorffDistance}(c), as
$$
\mathfrak{d}(I_{2},I_{3})\leq\norm{\Omega}\int_{0}^{t}\norm{B-\e^{sA}}\mathrm{d}s.
$$
Moreover, using the continuous differentiability of $\exp((\cdot)A)$,  the Ostrowski inequality in \cite[Theorem~1]{BarnettBuseCeroneDragomir02}, and  estimate \eqref{eq:BoundonExpAndL}, we have
$$\norm{B-\e^{sA}}\leq {t\norm{A}\e^{t\norm{A}}},~s\in \intcc{0,t}.
$$
Hence,
\begin{equation}\label{eq:dI2I3}
\mathfrak{d}(I_{2},I_{3})%&\leq \norm{\Omega}\int_{0}^{t}\norm{B-\e^{sA}}\mathrm{d}s\\
 \leq  \norm{\Omega}\int_{0}^{t}{t\norm{A}\e^{t\norm{A}}} \mathrm{d}s\leq t^2\norm{\Omega}\norm{A}\e^{t\norm{A}}.
\end{equation}
By combining the bounds \eqref{eq:dI1I3TriangleInequality}, \eqref{eq:dI1I2}, and \eqref{eq:dI2I3},  we get 
\begin{align*}
\mathfrak{d}(I_{1},I_{3})&\leq  2((1-\epsilon)t+t^2\norm{A})\e^{t\norm{A}}\norm{\Omega}\\
&\leq 2((1-\epsilon)t+t^2\norm{A})\e^{T\norm{A}}\norm{\Omega}.
\end{align*}
\end{proof}

\begin{proof}[Proof of Lemma \ref{lem:BoundingSandV}]
It can be shown using induction that $S_{i}^{N}=\exp(i\tau A)X_{0}$ and $V_{i}^{N}=\exp(i\tau A)\int_{0}^{\tau}\exp(sA)U\mathrm{d}s$ for all $i\in \intcc{0;N}$. Hence,
$$
\norm{S_{i}^{N}}\leq \e^{i\tau \norm{A}}\norm{X_{0}} \leq \e^{T\norm{A}}\norm{X_{0}}
$$
and
$$
\norm{V_{i}^{N}}\leq \e^{i\tau \norm{A}}\int_{0}^{\tau}\e^{s\norm{A}} \mathrm{d}s \norm{U} \leq %\exp(T \norm{A})\int_{0}^{T}\exp(T \norm{A}) \mathrm{d}s \norm{U}=
\tau\e^{2T\norm{A}}\norm{U}
$$
for all $i\in \intcc{0;N}$.
\end{proof}
\end{appendix}
\bibliographystyle{IEEEtran}

\end{document}